\documentclass[prx, aps,twocolumn,superscriptaddress,floatfix, 10pt]{revtex4-2}
\usepackage{amsthm}
\usepackage{graphicx}
\usepackage{physics}
\usepackage{amsmath,amssymb,amsthm}
\usepackage{graphicx}
\usepackage{stmaryrd}
\usepackage{booktabs}
\usepackage[table]{xcolor}
\usepackage{makecell}
\usepackage{array}
\usepackage[linesnumbered, ruled, vlined]{algorithm2e}

\theoremstyle{plain}
\newtheorem{theorem}{Theorem}

\newtheorem{corollary}[theorem]{Corollary}
\newtheorem{proposition}[theorem]{Proposition}
\newtheorem{example}{Example}

\usepackage{enumitem}

\theoremstyle{definition}
\newtheorem{definition}[theorem]{Definition}

\theoremstyle{remark}
\newtheorem{remark}[theorem]{Remark}

\theoremstyle{remark}

\usepackage{bm, physics}% bold math
\usepackage{tikz-cd}
\usepackage{mathtools}

\usepackage{color}

\newcommand\Z{\mathbb{Z}}
\newcommand\F{\mathbb{F}}

\renewcommand\Im{\operatorname{im}}
\renewcommand\rank{\operatorname{rank}}

\let\svthefootnote\thefootnote
\newcommand\freefootnote[1]{%
  \let\thefootnote\relax%
  \footnotetext{#1}%
  \let\thefootnote\svthefootnote%
}
\usepackage[colorlinks=true,hyperindex,breaklinks,
            linkcolor=blue,urlcolor=blue,citecolor=blue]{hyperref}
            
\begin{document}

% Or maybe Extending?
\title{Lifting Lifted Product Codes}

\author{Yuta Hirasaki}
\email{yutah2@illinois.edu}
\affiliation{Department of Physics, Illinois Quantum Information Science and Technology Center, and Institute of Condensed Matter Theory, The Grainger College of Engineering, University of Illinois at Urbana-Champaign, Urbana, Illinois 61801, USA}

\author{Jong Yeon Lee}
\email{jongyeon@illinois.edu}
\affiliation{Department of Physics, Illinois Quantum Information Science and Technology Center, and Institute of Condensed Matter Theory, The Grainger College of Engineering, University of Illinois at Urbana-Champaign, Urbana, Illinois 61801, USA}
\affiliation{Korea Institute for Advanced Study, Seoul 02455, South Korea}

\date{\today}

\begin{abstract}
    Lifted product (LP) codes form an important class of quantum error correcting codes with favorable code parameters. We introduce a systematic construction of LP code families based on group extensions and graph lifts. The construction increases the code size while preserving the local structure of the Tanner graph, and relates code parameters, logical operators, and fault-tolerant logical-operation gadgets within the families through chain and cochain maps.

    As a first application, we obtain LP codes with better code parameters than previously reported ones. 
    We then demonstrate that code-surgery gadgets can be transferred across the selected finite lifts through chain maps and, in several cases, implemented with lower space overhead. 
    We also develop parallel product surgery for \emph{lifted} clustered cyclic codes.
    Finally, we propose lifting as a systematic first step toward defining thermodynamic families for algebraically defined qLDPC codes without an underlying Euclidean lattice. For several base codes and selected lifts, coherent information exhibits finite-size crossings, while our results also indicate that additional conditions are needed to determine a unique family.
\end{abstract}

\pacs{}% insert suggested PACS numbers in braces on next line

\maketitle %\maketitle must follow title, authors, abstract and \pacs

\section{Introduction}

Quantum information is fragile against noise and extensive studies have been conducted for quantum error correcting (QEC) codes~\cite{CalderbankShor1996, Steane1996}. In particular, quantum low-density parity check (qLDPC) codes~\cite{PRXQuantum.2.040101} have been central to these studies. Among them, lifted product (LP) codes~\cite{Panteleev2021, Panteleev2022, Panteleev2021a} are an important class of qLDPC codes as they give a sequence of asymptotically good qLDPC codes~\cite{Panteleev2021a}, as well as finite-size instances with good performance~\cite{Xu2024, cain2026shorsalgorithmpossible10000, Bravyi2024}, such as quasi-cyclic~\cite{Xu2024, Panteleev2022}, bivariate bicycle (BB) codes~\cite{Bravyi2024, rmy6-9n89}, and many others~\cite{67xf-zdjb,wang2023abeliannonabelianquantumtwoblock, PhysRevA.109.022407, PhysRevA.88.012311}.

\begin{figure}[!t]
    \centering
    \includegraphics[width=1.0\linewidth]{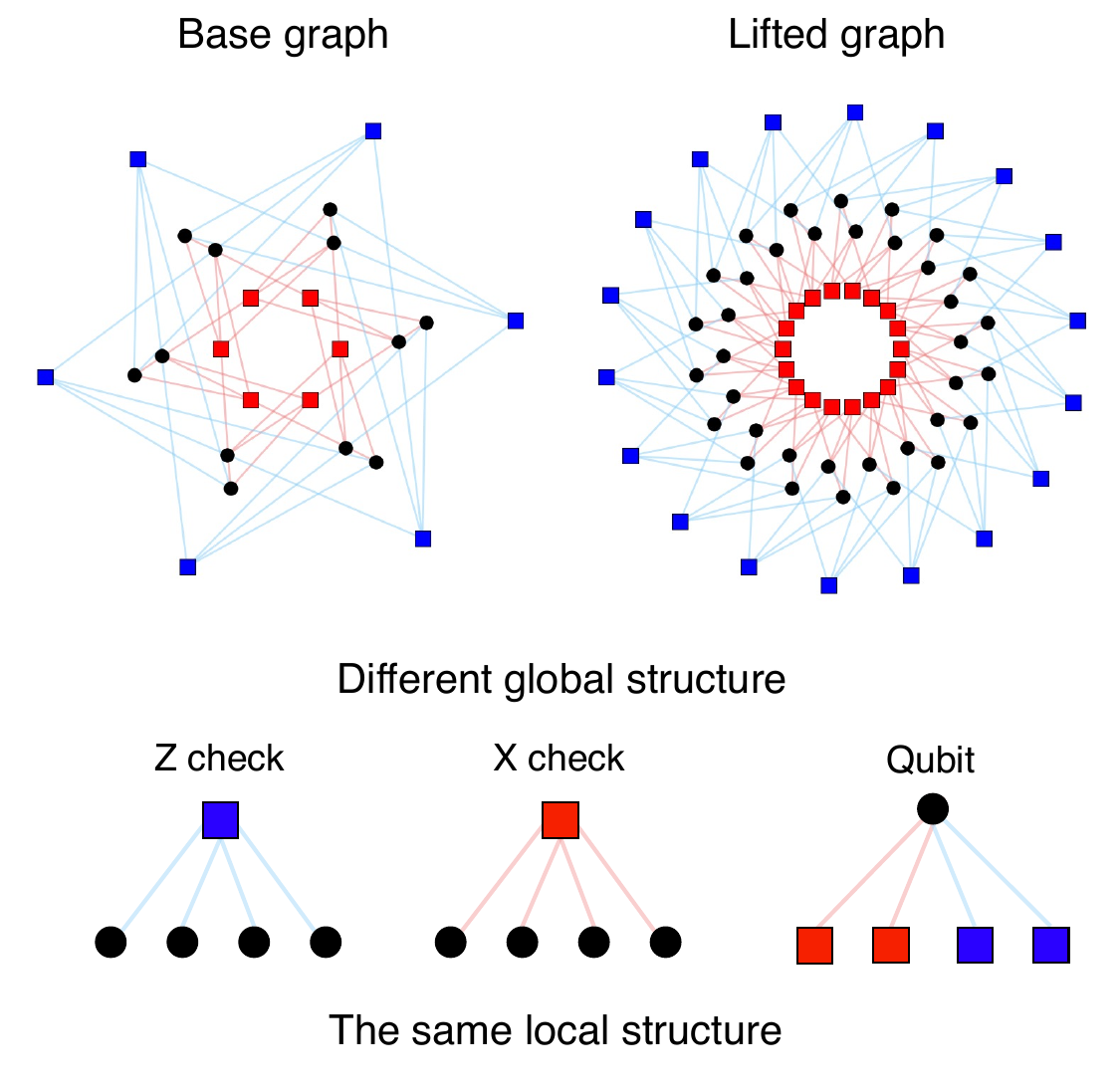}
    \caption{\textbf{Schematic picture of a lift of a Tanner graph.} The graphs show the Tanner graphs of a base and a lifted LP code. Black circles denote qubits, while blue and red squares denote $Z$ and $X$ checks, respectively. Although the two Tanner graphs have different global structures, they share the same local structure. The base (lifted) LP code is a cyclic code defined by the cyclic group $\Z_6$ 
    $(\Z_{18})$, respectively. The boundary maps are given by $H_A = \begin{pmatrix}
        x + x^2
    \end{pmatrix}$ and $H_B = \begin{pmatrix}
        x + x^{-1}
    \end{pmatrix}$ in both cases. The code parameters are $\llbracket12, 2, 3\rrbracket$ and $\llbracket36, 2, 3\rrbracket$, respectively.}
    \label{fig:LiftSchematicPic}
\end{figure}

In this work, we provide a systematic method for generating a sequence of LP codes from a given base LP code such that their Tanner graphs~\cite{Tanner1981} preserve the local incidence structure of the base Tanner graph. Put differently, we develop a framework for further lifting LP codes; see Fig.~\ref{fig:LiftSchematicPic} for an illustration. Although LP codes were originally introduced~\cite{Panteleev2022,Panteleev2021a,Panteleev2021} as lifts of hypergraph product codes~\cite{TillichZemor2014,PRXQuantum.2.040101}, explicit methods for further lifting a general LP code have not been developed beyond abstract constructions for arbitrary CSS codes~\cite{guemard2024liftsquantumcsscodes,guemard2025liftingcsscodehandlebody} or constructions restricted to particular subfamilies of LP codes~\cite{symons2025sequences}. Our approach is purely algebraic, and can be applied for any LP code.

As applications, we first show that our lifting protocol gives a systematic route to code search and can produce codes with improved parameters.
Sequences of BB codes obtained by lifting a BB code were previously studied in Ref.~\cite{symons2025sequences}, but that approach only explores the instances that remain within the BB family. In contrast, our framework applies to LP codes more generally, with BB codes appearing as a special case. 
Using this broader approach, we find lifted BB codes with improved parameters.
We further apply the same lifting procedure to recently proposed clustered cyclic (CC) codes~\cite{gu2026qgpuparallellogicquantum}, obtaining instances with enhanced parameters. 
We also discuss bounds on the parameters of lifted codes; in particular, in several settings the parameters of the lift are lower bounded by those of the base code~\cite{symons2025sequences}. Consequently, lifting a base code with favorable parameters provides a useful ansatz for identifying improved codes within a reduced search space.

Furthermore, we discuss the logical operation gadgets for lifted LP codes. Specifically, we show that the code surgery gadgets~\cite{Horsman2012LatticeSurgery, CowtanBurton2024, Cowtan2024SSIP, zheng2025highratesurgeryconstantoverheadlogical, Williamson2026, cross2025improvedqldpcsurgerylogical,yuan2026parsimoniousquantumlowdensityparitycheck, he2025extractorsqldpcarchitecturesefficient, PhysRevLett.134.070602, baspin2025fastsurgeryquantumldpc, gj8x-n5gg, doi:10.1126/sciadv.abn1717, PhysRevX.15.021088} are transferred  across the lifted family through chain maps induced by a lift, and we demonstrate that one can construct a surgery ancilla code with a smaller space overhead by utilizing the lifted structure of LP codes. We also identify the conditions under which lifted CC codes retain the same logical addressability under parallel product surgery as their base CC codes~\cite{gu2026qgpuparallellogicquantum}. We show that the lifted CC codes reported here preserve the relevant logical-operation capabilities.

Finally, we investigate decoherence-induced phase transitions~\cite{Lee2025symmetryprotected, LeePRXQ2023, PRXQuantum.5.020343, hlfh-86yz,PhysRevA.111.032402, PhysRevResearch.6.L042014, fx56-8nvy, vijay2025informationcriticalphasesdecoherence, yangPRX2026} of lifted LP codes. For qLDPC codes defined by abstract chain complexes or general graphs, it is often unclear how to construct a \emph{thermodynamic family}; a sequence of increasing system sizes that exhibits the same asymptotic critical behavior. We address this issue using lifts, which provide a sequence of increasingly large codes that preserve the local structure of a base code. We calculate the coherent information~\cite{PhysRevA.54.2629,Horodecki2006} under bit-flip noise for the lifted LP codes and observe finite-size crossings, suggesting that the selected lifted codes exhibit common asymptotic behavior. We also discuss the caveats of using local structure alone to define a thermodynamic family.

\vspace{5pt} {\bf Relation to prior works.} 
Recently, Guemard introduced a general notion of lifting CSS codes by associating to any CSS code a two-dimensional \emph{Tanner cone-complex} and defining lifted codes through finite coverings of this complex~\cite{guemard2024liftsquantumcsscodes}. Connected finite lifts are classified by finite-index subgroups of the fundamental group. While this provides a conceptually general classification, systematic code search requires determining the fundamental group and exploring its finite-index subgroups, which can become computationally difficult for generic codes. 
In addition, the proposed construction does not necessarily result in an LP code even when lifting an LP code. 
A lifting construction tailored to quantum Tanner codes was subsequently developed by taking coverings of the underlying square complex~\cite{guemard2025moderate}.

Indeed, for a CSS code with additional structures, finding a lift can be easier. In Ref.~\cite{symons2025sequences}, given a BB code, the authors find a way to construct a sequence of BB codes that cover the base BB code. They give a condition under which a BB code becomes a cover of another BB code, and search for lifted BB codes. Recently, Ref.~\cite{aydin2026breakingbicycleframecosetbased} has generalized this to two-block group algebra codes.

Our work moves one step further by finding a way to systematically lift any LP code over an \emph{arbitrary group}, which includes BB codes, via group extension. In this general construction, BB-preserving lifts are a special case, and even if we choose our base code to be a BB code, we can obtain a non-BB `further' lifted code with code parameters improved compared to BB family.

\vspace{5pt}
This paper is organized as follows. In Sec.~\ref{Sec:LiftingLPcodes}, we develop an algebraic framework for lifting LP codes and discuss the induced chain and cochain maps, together with bounds on the parameters of the lifted codes. In Sec.~\ref{Sec:CodeSearch}, we present finite-size examples of lifted LP codes obtained using our framework. In Sec.~\ref{Sec:LogicalGadget}, we study logical-operation gadgets for lifted LP codes and show, in particular, that the lifted structure can reduce the space overhead of code surgery. In Sec.~\ref{Sec:CI}, we argue that lifts can be used to construct thermodynamic families of LP codes, and demonstrate finite-size crossings of coherent information under decoherence. Finally, in Sec.~\ref{Sec:conclusion_outlook}, we summarize our results and discuss future directions.

% \tableofcontents

\section{Lift, Projection, and Chain/cochain maps}\label{Sec:LiftingLPcodes}

A graph lift enlarges a graph while preserving the local incidence structure around every vertex. Formally, a graph $\mathcal{G}'$ is a lift of $\mathcal{G}$ if there is a surjective graph map $p:\mathcal{G}'\to\mathcal{G}$ that induces a bijection between the edges incident to $v'$ and those incident to $p(v')$ for every $v'\in V(\mathcal{G}')$~\footnote{More precisely, the neighborhood of $v$ should refer to the set of \emph{edges} incident to $v$, not the set of vertices in general. This distinction matters when multiple edges connect the same pair of vertices. For Tanner graphs considered here, which are simple undirected graphs, the two descriptions are equivalent. We use the vertex-based definition in the main text because it provides a more intuitive interpretation of the thermodynamic discussion in Sec.~\ref{Sec:CI}.}. If every vertex of $\mathcal{G}$ has exactly $t$ preimages, then $\mathcal{G}'$ is called a $t$-lift of $\mathcal{G}$. As illustrated in Fig.~\ref{fig:LiftSchematicPic}, a lift preserves the local Tanner-graph incidence structure while allowing its global structure to change.

Our goal is to provide a systematic way to produce a sequence of LP codes from a given base LP code by lift. We take an algebraic approach, and we obtain admissible lifts by exploiting a group extension
\begin{align}\label{Eq:group_extension}
    1\longrightarrow K\xlongrightarrow{\chi} H\xlongrightarrow{\pi}G\longrightarrow 1.
\end{align}
Starting from an LP code over $\F_2[G]$, we lift the two underlying classical complexes to $\F_2[H]$ and then take their balanced product. The resulting LP code has a Tanner graph~\cite{Tanner1981} that is a $|K|$-lift of the original Tanner graph. Conversely, the homomorphism $\pi$ projects an LP code over $\F_2[H]$ to a smaller LP code over $\F_2[G]$. These constructions are summarized in Fig.~\ref{Fig:LiftbySES}. The same projection also induces chain and cochain maps relating the two codes.

In Sec.~\ref{Subsec:LiftingLPcodes}, we construct the forward lift and prove that it gives a lift of a Tanner graph. In Sec.~\ref{Subsec:projectingLPcodes}, we define the inverse projection to a quotient group. Finally, in Sec.~\ref{Subsec:ChainCochainMaps}, we discuss the induced transfer maps and use them to relate logical operators and code parameters across the lifted family. These maps will provide the algebraic mechanism for transferring logical-operation gadgets from a projected code to its lift in Sec.~\ref{Sec:LogicalGadget}.

\begin{figure}
    \centering
    \includegraphics[width=1.0\linewidth]{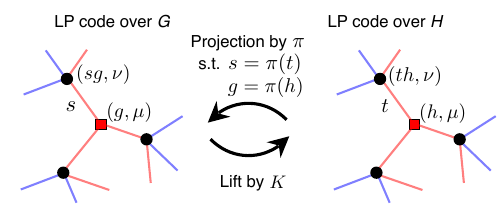}
    \caption{\textbf{Local structures of lift and projection.} The left (right) figure shows the local structure of the Tanner graph of an LP code over a group $G$ ($H$). The red squares denote $X$-checks, and the black dots denote the qubits. Each vertex is labeled by a pair of a group element and an index labeling a basis element of the corresponding module: for example $(g, \mu)$ in the figure, where $g\in G$ and $\mu$ indexes a basis element in $A_1\otimes B_1$. Edges are generated by the group elements with nonzero coefficients, i.e., \emph{supported group elements} in the boundary maps. Specifically, take the $(\mu, \nu)$-entry, and let $S_{\mu\nu}\subseteq G$ be the set of the supported group elements. An element $s\in S_{\mu\nu}$ connects $(g,\mu)$ to $(sg,\nu)$ for every $g\in G$. The Tanner graph of the LP code over $H$ is described analogously using edge generating sets $T_{\mu\nu}$. When the homomorphism $\pi$ in Eq.~\eqref{Eq:group_extension} satisfies the relation $s = \pi(t)$ for every edge generator $s$ and $t$, $\pi$ induces a projection map between the two Tanner graphs and the LP code over $H$ becomes a lift of the LP code over $G$.}
    \label{Fig:LiftbySES}
\end{figure}

\subsection{Lift via group extension}\label{Subsec:LiftingLPcodes}

 Consider an LP code $Q$ over the group algebra $R_G \coloneqq \F_2[G]$, constructed from a pair of classical codes $A_\bullet$ and $B_\bullet$. 
\begin{align}\label{Eq:LPClassicalCodes}
    A_1\xlongrightarrow{\partial_A}A_0, \quad B_1\xlongrightarrow{\partial_B} B_0,
\end{align}
where $A_1, A_0, B_1, B_0$ are free modules over $R_G$. 
For $A_i \simeq R_G^{n_i^A}$, we call $n_i^A$ the \emph{rank} of $A_i$.
When $G$ is non-abelian, $R_G$ is non-commutative, so the handedness of the module structures must be specified. We take ${A}_0,{A}_1$ to be free right $R_G$-modules and ${B}_0,{B}_1$ to be free left $R_G$-modules. Accordingly, the boundary maps are $R_G$-linear in the corresponding sense
\begin{align}
{\partial}_A(a r)={\partial}_A(a)r,
\qquad
{\partial}_B(r b)=r{\partial}_B(b),
\end{align}
for $r\in R_G$, $a\in A_1$, and $b\in B_1$. Taking the balanced product~\cite{9490244} of these two chain complexes, we obtain the LP code $Q$ described by the following length-two chain complex:
\begin{equation}
    \begin{tikzcd}
    & A_1\otimes B_0 & \\
    {A_1\otimes B_1} & & A_0\otimes B_0\\
     & A_0\otimes B_1 &
    \arrow["{I\otimes \partial_B}", from=2-1, to=1-2]
    \arrow["{\partial_A\otimes I}", from=1-2, to=2-3]
     \arrow["{\partial_A\otimes I}", from=2-1, to=3-2]
     \arrow["{I \otimes \partial_B}", from=3-2, to=2-3]
    \end{tikzcd}.
\end{equation}

We lift $Q$ using a group extension. Suppose that $G$ is extended by a kernel group $K$ to a group $H$, with the homomorphisms $\chi$ and $\pi$ defined by the short exact sequence in Eq.~\eqref{Eq:group_extension}. Note that the homomorphism $\pi$ is surjective. For each fiber $\{ h \,|\, \pi(h) = g \}$, we choose a representative of $g$ in $H$, denoted by the \emph{section} $\sigma: G\to H$ such that
\begin{align}
    (\pi\circ\sigma)(g) = g,\quad \forall g\in G.
\end{align}
Here $\sigma$ need not be a group homomorphism. With the choice of $\sigma$, one can construct the following map
\begin{align}\label{Eq:BijectionofSES}
    F: K\times G\to H;\;\;(\gamma, g) \mapsto \chi(\gamma)\sigma(g),
\end{align}
which is bijective. For surjectivity, for any $h \in H$ with $\pi(h) = g$, note that $\pi(h \sigma(g)^{-1}) = e_G$. Therefore, $h \sigma(g)^{-1} \in \ker \pi = \textrm{im} \chi$. Hence there exists $\gamma \in K$ such that $\chi(\gamma) = h \sigma(g)^{-1}$, i.e., $h = \chi(\gamma) \sigma(g)$. For injectivity, suppose $\chi(\gamma) \sigma(g) = \chi(\gamma') \sigma(g')$. Applying the group homomorphism $\pi$, we obtain $g = g'$. This indicates that $\chi(\gamma) = \chi(\gamma')$. However, since $\chi$ is injective, $\gamma = \gamma'$. 

To proceed, we first lift the base classical chain complexes $A_\bullet$ and $B_\bullet$ over $R_G$ to chain complexes $\tilde{A}_\bullet$ and $\tilde{B}_\bullet$ over $R_H \coloneqq \F_2[H]$:
\begin{align}\label{Eq:LiftofClassicalCodes}
\widetilde{A}_1\xlongrightarrow{\widetilde{\partial}_A}\widetilde{A}_0, \quad \widetilde{B}_1\xlongrightarrow{\widetilde{\partial}_B} \widetilde{B}_0.
\end{align}
When $H$ is non-abelian, the handedness of the modules is defined in the same way as $Q$. The boundary maps are defined by the following procedure.

Take the $(i, j)$-entry of the boundary maps $\partial_A$:
\begin{align}
    (\partial_A)_{ij} = \sum_{g\in G}a_{ij, g} g, \quad a_{ij, g} \in \mathbb{F}_2,
\end{align}
where $i \in [\rank_{R_G} A_0]$ and $j\in[\rank_{R_G} A_1]$. Here $[n] = \{0, 1, \dots, n-1\}$ for a positive integer $n$. For each group element $g$ with $a_{ij, g} = 1$ , the fiber $\pi^{-1}(g)$ contains exactly $|K|$ elements. We choose one such element and denote the chosen lift of $g$ at position $(i,j)$ by $h_{ij}(g)\in H$. Given $(\chi,\pi,\sigma)$, this choice is uniquely specified by an element $\gamma_{ij}(g)\in K$ through
\begin{align}
    h_{ij}(g) := F(\gamma_{ij}(g), g) =  \chi(\gamma_{ij}(g)) \sigma(g).
\end{align}
The lifted boundary map $\widetilde{\partial}_A$ is then given by
\begin{align}\label{Eq:LiftedBoundaryMapofA}
    (\widetilde{\partial}_A)_{ij} = \sum_{g\in G}a_{ij, g}\chi(\gamma_{ij}(g))\sigma(g)\in R_H,
\end{align}
for $i\in [\rank_{R_G} A_0]$ and $j\in [\rank_{R_G} A_1]$. Note that $\rank_{R_H}\widetilde{A}_0 = \rank_{R_G} A_0$ and $\rank_{R_H}\widetilde{A}_1 = \rank_{R_G} A_1$ by construction. The classical code $B_\bullet$ is lifted analogously.

In this way, we obtain $\widetilde{A}_\bullet$ and $\widetilde{B}_\bullet$. By taking the balanced product of these two over $R_H$, we obtain an LP code $\widetilde{Q}$. One can show that $\widetilde{Q}$ obtained is a $|K|$-lift of $Q$. 
\begin{proposition}[Lift of an LP code]\label{theorem:liftingLPbyGroupExtensions}
    The Tanner graph of the LP code $\widetilde{Q}$ obtained above is a $|K|$-lift of the Tanner graph of the original LP code $Q$. 
\end{proposition}

\begin{proof}
    The complete proof is provided in Appendix~\ref{Sec:ProvingLiftingTheorem}, so we illustrate the idea of the proof here. We first show that the homomorphism $\pi$ induces a graph projection. The preimages of each vertex are given by $K$ from the first isomorphism theorem. The graph projection is locally bijective from the construction of $\widetilde{\partial}_A$ and $\widetilde{\partial}_B$. Therefore, the Tanner graph of $\tilde{Q}$ is a $|K|$-lift of the Tanner graph of $Q$. 
\end{proof}

\noindent Therefore, the choices entering a further lift in our framework can be summarized as follows:
\begin{itemize}[leftmargin=13pt]
    \item A group extension
    \[
        1\longrightarrow K\xrightarrow{\chi}H\xrightarrow{\pi}G
        \longrightarrow 1,
    \]
    considered up to equivalence. Its lift index is $|K|$. For fixed $\pi$, changing $\chi$ only relabels the kernel. Note that $\pi$ is not unique for a given $(K, H, G)$ in general.

    \item For each supported element $g$ in every entry of the boundary maps of $A$ and $B$, a lift
    \[
        h_{ij}(g)\in\pi^{-1}(g).
    \]
    Given a section $\sigma\colon G\to H$, this choice is parametrized by
    $h_{ij}(g)=\chi(\gamma_{ij}(g))\sigma(g)$; the choice of $\sigma$ is gauge. Inequivalent choices of the elements $h_{ij}(g)$ can produce distinct LP codes even for a fixed extension.
\end{itemize}

\begin{example}[Lift of hypergraph product codes~\cite{TillichZemor2014}]
    We consider a group $H$ and a trivial group $G\,{=}\,\{e\}$, along with the following exact sequence
    \begin{align}
        1\longrightarrow H \xlongrightarrow{\chi} H \xlongrightarrow{\pi} 1 \longrightarrow 1,
    \end{align}
    where $\chi(h) = h$ and $\pi(h) = e$ for any $h\in H$. This gives a lift of hypergraph product codes~\cite{Panteleev2021, Panteleev2022, Panteleev2021a}.
\end{example}

\begin{example}[Lift of BB codes in Ref.~\cite{symons2025sequences}]
    We consider a group $G = \Z_l\times \Z_m$, $K = \Z_u\times \Z_t$ and $H = \Z_{ul}\times \Z_{tm}$ that fit into the following exact sequence
    \begin{align}
        1\longrightarrow \Z_u\times \Z_t \xlongrightarrow{\chi} \Z_{ul}\times \Z_{tm} \xlongrightarrow{\pi} \Z_l\times \Z_m \longrightarrow 1,
    \end{align}
    This is equivalent to a sequence of BB codes in Ref.~\cite{symons2025sequences}. We note that depending on whether $(u, l)$ and $(t, m)$ are coprime or not, the above extension can be a trivial extension (direct product) or non-trivial extension (classified by the group cohomology $H^2(G, K)$).
\end{example}

\begin{example}[Non-Abelian Group]\label{Ex:NonAbelianLP}
    We consider the group $G = S_3\times S_3$, and the boundary maps
    \begin{align}
        a = 1 + s_1 s_2r_2 + r_2s_2, \quad b = 1 + r_1 + r_2,
    \end{align}
    where we use the presentation $S_3 = \langle r, s |  r^3 = s^2 = rsrs = 1\rangle $, and the subscript $1 (2)$ means the generators for the first (second) group $S_3$ in the direct product $G = S_3\times S_3$.
    The base code has the code parameter $\llbracket 72, 4, 6\rrbracket$. 
    We lift this code using the groups $K = \Z_2, \Z_3, \Z_4$, and the lifted codes are shown in Table~\ref{Table:LiftedLPforCI}. We study these codes in Sec.~\ref{Sec:CI}.
\end{example}

\subsection{Projection via quotient map}\label{Subsec:projectingLPcodes}

We next consider the reverse direction. Given an LP code $Q$ over $R_H$ and a quotient homomorphism $\pi:H\to G$, we construct a smaller LP code $Q_{\pi}$ over $R_G$ by applying $\pi$ to every group-algebra entry of the boundary maps. Again, see Fig.~\ref{Fig:LiftbySES} for an illustration.

Suppose that $Q$ is constructed from classical boundary maps $\partial_A$ and $\partial_B$ over $R_H$ in Eq.~\eqref{Eq:LPClassicalCodes}, and that $H$ fits into the short exact sequence in Eq.~\eqref{Eq:group_extension} for some $K$ and $G$. Then the homomorphism $\pi$ induces the ring (group-algebra) homomorphism
\begin{align}
\Pi:R_H&\longrightarrow R_G,\\
\sum_{h\in H}a_hh&\longmapsto\sum_{h\in H}a_h\pi(h),
\end{align}
where $a_h\in \F_2$. Applying $\Pi$ entrywise to $\partial_A$ gives the projected boundary map $\partial_A^\pi$, whose entries are
\begin{align}
    (\partial_A^\pi)_{ij} \coloneqq \Pi((\partial_A)_{ij}) = \sum_{h\in H}a_{ij, h}\pi(h).
\end{align}
We define $\partial_B^\pi$ from $\partial_B$ in the same way. Taking the balanced product over $R_G$ of the length-one chain complexes defined by $\partial_A^\pi$ and $\partial_B^\pi$ yields an LP code $Q_\pi$. We call $Q_\pi$ the \emph{projected code} of $Q$ under $\pi$.

The ring homomorphism $\Pi$ also induces projection maps on the free modules $A_1$, $A_0$, $B_1$, and $B_0$. For example, since $A_1$ is a free right $R_H$-module, every element $a_1\in A_1$ can be written as
\begin{align}
    a_1 = \sum_{i=1}^{n_1^A}e_i r_i,
\qquad
r_i\in R_H,
\end{align}
where $\{e_i\}_{i=1}^{n_1^A}$ is a basis of $A_1$. Applying $\Pi$ to each coefficient defines the projection
\begin{align}
    &p_{A_1}:A_1\longrightarrow A_1^\pi,\\
    &p_{A_1}(a_1) = \sum_{i=1}^{n_1^A}e_i\Pi(r_i) \in \left(R_G\right)^{n_1^A} \simeq A_1^\pi.
\end{align}
The projection maps on $A_0$, $B_1$, and $B_0$ are defined analogously.

Using the natural isomorphism of $\F_2$-vector spaces
\begin{align}\label{Eq:LineraIsomorphism}
    A_{i}\otimes_{R_H} B_j\simeq \F_2^{|H|n_{i}^An_{j}^B}, \quad A_{i}^\pi\otimes_{R_G} B_j^\pi\simeq \F_2^{|G|n_{i}^An_{j}^B}
\end{align}
for $i = 0, 1$ and $j = 0, 1$, the module projections defined above induce $\F_2$-linear maps between the chain complexes $Q$ and $Q_\pi$. These maps commute with the corresponding boundary maps and therefore define chain maps between the two codes, as stated formally below.
\begin{proposition}[Chain/cochain maps by projection]\label{theorem:projectingLPbyGroupExtensions}
    The induced chain maps yield the following commuting diagram:
    \begin{equation}\label{Eq:LPchain}
        \begin{tikzcd}
        {{Q}_2} & {{Q}_1} & {Q_0} \\
    	{Q_{2}^{\pi}} & {Q_{1}^{\pi}} & {Q_{0}^{\pi}}
    	\arrow["{\partial_2}", from=1-1, to=1-2]
    	\arrow["{{\partial}_1}", from=1-2, to=1-3]
            \arrow["{{\partial}_2^{\pi}}", from=2-1, to=2-2]
    	\arrow["{\partial_1^{\pi}}", from=2-2, to=2-3]
            \arrow["{p_2}", from=1-1, to=2-1]
    	\arrow["p_1", from=1-2, to=2-2]
            \arrow["p_0", from=1-3, to=2-3]
        \end{tikzcd},
    \end{equation}
    where $Q_\bullet\; (Q_\bullet^\pi)$ is a chain complex over $\F_2$ for the code $Q\;(Q_\pi)$, respectively. The chain maps are given as 
    \begin{align}
        p_2 &= p_{A_1}\otimes p_{B_1}\\
        p_1 &= \left(p_{A_1}\otimes p_{B_0}\right)\oplus \left(p_{A_0}\otimes p_{B_1}\right)\\
        p_0 &= p_{A_0}\otimes p_{B_0}.
    \end{align}
    These maps are well defined on the balanced tensor products because,
    for every $r\in R_H$,
    \begin{align}
        p_{A_i}(ar)
        &=
        p_{A_i}(a)\Pi(r),
        &
        p_{B_j}(rb)
        &=
        \Pi(r)p_{B_j}(b).
    \end{align}
    
    The induced cochain maps yield the following commuting diagram.
    \begin{equation}\label{Eq:LPcochain}
        \begin{tikzcd}
        {{Q}^0} & {{Q}^1} & {Q^2} \\
    	{Q^{0}_{\pi}} & {Q^{1}_{\pi}} & {Q^{2}_{\pi}}
    	\arrow["{\partial^1}", from=1-1, to=1-2]
    	\arrow["{{\partial}^2}", from=1-2, to=1-3]
            \arrow["{{\partial}^1_{\pi}}", from=2-1, to=2-2]
    	\arrow["{\partial^2_{\pi}}", from=2-2, to=2-3]
            \arrow["{p^0}", from=1-1, to=2-1]
    	\arrow["p^1", from=1-2, to=2-2]
            \arrow["p^2", from=1-3, to=2-3]
        \end{tikzcd},
    \end{equation}
    where $Q^\bullet\; (Q^\bullet_\pi)$ is a cochain complex for the code $Q\;(Q_\pi)$, respectively. Here, $Q^i$ is a dual space of $Q_i$.
    Note that this is not the canonical pullback of the chain map.
    
    Moreover, when each entry of the boundary maps is not degenerate under $\Pi$, i.e., no two distinct group elements in its support have the same image under $\pi$, the code $Q$ is an index-$|K|$ cover of $Q_{\pi}$. 
\end{proposition}

\begin{proof}
    We illustrate the idea of the proof here, and give a detailed discussion in Appendix~\ref{Sec:ProvingProjectionTheorem}. We first show that the chain diagram commutes. It suffices to show 
    \begin{align}
        p_1\circ\partial_2 (e^A_{i}h\otimes e^B_j) = \partial_2^\pi\circ p_2(e^A_{i}h\otimes e^B_j),
    \end{align}
    for every $h\in H$, $i \in [n_{1}^A]$, and $j\in [n_1^B]$. This follows directly from the definition of the projected boundary maps and the fact that $\Pi$ is a ring homomorphism. The same calculation gives $p_0\circ\partial_1=\partial_1^\pi\circ p_1$
    Hence, $p_\bullet$ defines a chain map.
        
    The cochain statement is slightly more subtle. Indeed, the coboundary matrices are obtained by transposing the boundary matrices and applying the group-algebra involution to each entry. It follows, by an analogous calculation on basis cochains, that
    \begin{align}
        p^1\circ\partial^1
        \left(e_B^j\otimes h e_A^i\right)
        =
        \partial_\pi^1\circ p^0
        \left(e_B^j\otimes h e_A^i\right),
    \end{align}
    for every $h\in H$, $i\in[n_1^A]$, and $j\in[n_1^B]$. Here $e^i_A$ $(e^j_B)$ is the $R_H$-dual basis of $e_i^A$ $(e_j^B)$, respectively.  The remaining cochain-map identity follows in the same way.
    
    We prove the last statement. Suppose that the boundary maps are nondegenerate under $\Pi$. Consider an entry
    \begin{align}
        (\partial_A^\pi)_{ij} = \sum_{h\in H}a_{ij,h}\pi(h).
    \end{align}
    Nondegeneracy implies that $\pi(h)$ defines a distinct element in $G$ for $h\in H$ with $a_{ij, h} = 1$. Since the map $F:K\times G\rightarrow H$ in Eq.~\eqref{Eq:BijectionofSES} is bijective, there is a unique $\gamma_{ij}(\pi(h))$ such that 
    \begin{align}
        h = F\left(\gamma_{ij}(\pi(h)),\pi(h)\right).
    \end{align}
    Therefore, $\partial_A$ can be reconstructed from $\partial_A^\pi$ by assigning $\gamma_{ij}(\pi(h))$ to each supported element:
    \begin{align}
        (\partial_A)_{ij}
        =
        \sum_{h\in H}
        a_{ij,h}
        F\left(\gamma_{ij}(\pi(h)),\pi(h)\right).
    \end{align}
    The same argument applies to $\partial_B$. Hence, $Q$ is obtained from $Q_\pi$ through the lifting construction of Proposition~\ref{theorem:liftingLPbyGroupExtensions}. It follows that the Tanner graph of $Q$ is an index-$|K|$ cover of the Tanner graph of $Q_\pi$.
\end{proof}
If the boundary maps are degenerate under $\Pi$, degenerate terms may cancel under projection, thereby changing the check weights. In this case, the Tanner graph of $Q$ is not a graph lift of the standard Tanner graph of $Q_\pi$. Nevertheless, the projection still induces the chain and cochain maps described above.

\begin{remark}
    If the degenerate terms are retained after projection, they define a generalized Tanner graph where multiple edges may connect the same check and qubit vertices. The Tanner graph of $Q$ can then be regarded as a $|K|$-lift of this generalized Tanner multigraph. We note that such multigraphs are not the standard Tanner-graph convention for qLDPC codes~\cite{PRXQuantum.2.040101}, as parallel edges are combined modulo $2$ in the parity check matrices.
\end{remark}

\begin{example}[Gross code as a cover code]
    We study the projected codes of the gross code $\llbracket144, 12, 12\rrbracket$~\cite{Bravyi2024, yoder2025tourgrossmodularquantum}. The gross code is defined over a group $\F_2[\Z_{12}\times \Z_6] \simeq \F_2[x, y]/(x^{12} - 1, y^{6}- 1)$ with the boundary maps
    \begin{align}
        a = x^3 + y + y^2,\quad b = y^3 + x + x^2.
    \end{align}
    We consider the short exact sequence in Eq.~\eqref{Eq:group_extension} to obtain projected LP codes. This is essentially the factorization problem of $12$ and $6$. The results are summarized in Table~\ref{Table:GrossProjectedCodes}. Note that the boundary maps are not uniquely determined, as the homomorphism $\pi$ is not unique. Specifically, given a short exact sequence in Eq.~\eqref{Eq:group_extension} and a homomorphism $\pi$, a homomorphism $\psi\circ\pi$ also defines the same short exact sequence for an automorphism $\psi$ on $G$. We note that the projected code obtained from $\psi\circ\pi$ is isomorphic to the code obtained from $\pi$. 
\end{example}

\begin{table}
\centering

\setlength{\tabcolsep}{6pt}      % default is about 6pt
\renewcommand{\arraystretch}{1.15}

\begin{tabular}{c|c|c|c}
\toprule
Group $G$ & $H_a$ & $H_b$ & $\llbracket n,k,d\rrbracket $ \\
\midrule
$\Z_{12}\times\Z_3$ & $x^3 + y + y^2$ & $1 + x + x^2$ & $\llbracket 72, 8, 6\rrbracket$\\ % Q46
$\Z_{12}\times\Z_3$ & $x^3 + x^{10}y + x^8y^2$ & $x^6 + x + x^2$ & $\llbracket 72, 8, 6\rrbracket$\\ % Q42
$\Z_{6}\times\Z_6$ & $x^3 + y + y^2$ & $y^3 + x + x^2$ & $\llbracket 72, 12, 6\rrbracket$\\ % Q41
$\Z_{12}\times\Z_2$ & $x^3 + x^{10}y + x^8$ & $x^6y + x + x^2$ & $\llbracket 48, 4, 6\rrbracket$\\ % Q44
$\Z_{12}\times\Z_2$ & $x^3 + x^8y + x^4$ & $y + x + x^2$ & $\llbracket 48, 8, 4\rrbracket$\\ % Q45
$\Z_{6}\times\Z_3$ & $1 + x + x^2$ & $x^3 + y + y^2$ & $\llbracket 36, 8, 4\rrbracket$\\ % Q39
$\Z_{6}\times\Z_3$ & $x + x^2 + x^3$ & $x^3 + xy + x^2y^2$ & $\llbracket 36, 8, 4\rrbracket$\\ % Q40
$\Z_{6}\times\Z_3$ & $x^3 + y + y^2$ & $1 + x + x^2$ & $\llbracket 36, 8, 4\rrbracket$\\ % Q23
\bottomrule
\end{tabular}
\caption{
\textbf{Projected codes of the gross code.} We obtain the list of projected codes by identifying distinct group extensions in  Eq.~\eqref{Eq:group_extension} for $\mathbb{Z}_{12} \times\mathbb{Z}_6$. Note that the table is not exhaustive.
}
\label{Table:GrossProjectedCodes}
\end{table}

\subsection{Chain/cochain maps}\label{Subsec:ChainCochainMaps}
Before proceeding to the next section, we review several properties of the chain and cochain maps arising from lift and projection. In particular, we discuss the induced maps on homology and cohomology groups, or equivalently, how logical operators in the base and lifted codes are related. We explain how these properties depend on the lift index $|K|$. We also discuss bounds on the code parameters. Some of the discussion can be found in  Refs.~\cite{hatcher2002algebraic, symons2025sequences, guemard2025moderate}. 

Suppose that an LP code $Q$ is projected to an LP code $Q_{\pi}$ with the induced chain (cochain) maps in Eqs.~\eqref{Eq:LPchain} and~\eqref{Eq:LPcochain}. Taking the $\F_2$-dual of these projection maps gives maps in the opposite direction~\cite{hatcher2002algebraic}, which are described by the cochain complexes
\begin{equation}\label{Eq:LPchain_transfermap}
    \begin{tikzcd}
    {{Q}_2} & {{Q}_1} & {Q_0} \\
    {Q_{2}^{\pi}} & {Q_{1}^{\pi}} & {Q_{0}^{\pi}}
    \arrow["{\partial_2}", from=1-1, to=1-2]
    \arrow["{{\partial}_1}", from=1-2, to=1-3]
        \arrow["{{\partial}_2^{\pi}}", from=2-1, to=2-2]
    \arrow["{\partial_1^{\pi}}", from=2-2, to=2-3]
        \arrow["{\tau_2}", from=2-1, to=1-1]
    \arrow["\tau_1", from=2-2, to=1-2]
        \arrow["\tau_0", from=2-3, to=1-3]
    \end{tikzcd},
\end{equation}
and 
\begin{equation}\label{Eq:LPcochain_transfermap}
    \begin{tikzcd}
    {{Q}^0} & {{Q}^1} & {Q^2} \\
    {Q^{0}_{\pi}} & {Q^{1}_{\pi}} & {Q^{2}_{\pi}}
    \arrow["{\partial^1}", from=1-1, to=1-2]
    \arrow["{{\partial}^2}", from=1-2, to=1-3]
        \arrow["{{\partial}^1_{\pi}}", from=2-1, to=2-2]
    \arrow["{\partial^2_{\pi}}", from=2-2, to=2-3]
        \arrow["{\tau^0}", from=2-1, to=1-1]
    \arrow["\tau^1", from=2-2, to=1-2]
        \arrow["\tau^2", from=2-3, to=1-3]
    \end{tikzcd}.
\end{equation}
These canonically induced maps are called \emph{transfer maps}. We note that $\tau^i\;(\tau_i)$ is canonically induced from $p_i\;(p^i)$, respectively. At the matrix level, the transfer maps are obtained by taking the transpose of the corresponding projection maps. 

Equivalently, the transfer maps admit the following representation, which is useful to characterize logical operators in the lifted codes. Let $\{e_j^\pi\}_j$ and $\{e_k\}_k$ be the canonical group-element bases of $Q_i^\pi$ and $Q_i$, respectively.  With respect to these bases and their duals, the chain and cochain projections have the same binary matrix. Thus the chain transfer map is given by
\begin{align}\label{Eq:transfermap_def}
    \tau_i\left(e_j^\pi\right) = \sum_{\substack{k\\p_i(e_k)=e_j^\pi}} e_k,
\end{align}
and extended $\F_2$-linearly to all of $Q_i^\pi$. The cochain transfer map $\tau^i$ admits an analogous description.

It is well known that chain/cochain maps induce maps on the homology/cohomology groups~\cite{hatcher2002algebraic}, respectively. We denote the maps induced by projection on the $i$-th homology and cohomology groups by $\hat{p}_i$ and $\hat{p}^i$, and those induced by transfer by $\hat{\tau}_i$ and $\hat{\tau}^i$, respectively. In the physics language, applying the transfer map to a logical operator of the projected code $Q_{\pi}$ produces a well-defined logical class of the lifted code $Q$; the resulting class depends only on the logical equivalence class of the original operator and not on the choice of representative. Although the resulting class may be trivial, with a slight abuse of terminology, we refer to homology/cohomology classes obtained in this way as \emph{lifted logical operators}.

\begin{proposition}
    The pair of the lifted logical $X$ and $Z$ operators commutes if the lift index $t$ is even. When $t$ is odd, they anticommute iff the underlying logical operators anticommute in the projected code.
\end{proposition}

\begin{proof}
    This can be proved by the definition of the lifted logical operators and the transfer maps in Eq.~\eqref{Eq:transfermap_def}. Specifically, the number of intersections is multiplied by the lift index $t$ under the transfer map. When $t$ is even, the lifted logical operators always overlap an even number of times and thus they commute. When $t$ is odd, the number of intersections is unchanged modulo 2, and thus the statement holds.
\end{proof}

For odd-index lifts, we have the following proposition, which has been used to derive the bounds on the code parameters~\cite{symons2025sequences, guemard2025moderate}.

\begin{proposition}
    When $t$ is odd, the induced transfer maps on the homology and cohomology groups are injective in every degree. Moreover, the induced projection maps on these groups are surjective in every degree.
\end{proposition}

\begin{proof}
    The projection and transfer maps satisfy
    \begin{align}
        \hat{p}_i\circ\hat{\tau}_i &= tI =  I, & \hat{p}^i\circ\hat{\tau}^i &= tI = I.
    \end{align}
    Here, multiplication by $t$ is the identity over $\F_2$ because $t$ is odd.
    It follows that $\hat{\tau}_i$ and $\hat{\tau}^i$ are injective, while $\hat{p}_i$ and $\hat{p}^i$ are surjective. Indeed, if $\hat{\tau}_i(l)=0$, then applying $\hat{p}_i$ gives
    \begin{align}
        l = \hat{p}_i\circ\hat{\tau}_i(l) = \hat{p}_i(0) =0,
    \end{align}
    which proves injectivity. Conversely, for any $l\in H_i(Q_\pi)$, the class $\hat{\tau}_i(l)\in H_i(Q)$ satisfies
    \begin{align}
        \hat{p}_i\left(\hat{\tau}_i(l)\right) = l,
    \end{align}
    which proves surjectivity. The cohomological statements follow analogously.

    In particular, these conclusions hold for the first homology and cohomology groups, which describe the logical $X$ and $Z$ operators, respectively.
\end{proof}

Finally, we state the bounds on the code parameters for odd index lifts.
\begin{corollary}[Bounds on the lifted code parameters~\cite{symons2025sequences, guemard2025moderate}]
    When the lift index $t$ is odd, the number of logical qubits in the lifted code cannot be smaller than the number of logical qubits in the base code; i.e., $k'\geq k$ holds, where $k'\;(k)$ is the number of logical qubits in the lifted (base) code. Moreover, the distance of the lifted code is bounded as $d'\leq td$, where $d'\;(d)$ is the distance of the lifted (base) code. When $k' = k$, we also have $d'\geq d$.
\end{corollary}

\begin{proof}
    A detailed discussion is given in Ref.~\cite{symons2025sequences}, so we sketch the ideas of the proof here. The first bound on the number of logical qubits follows from the injectivity of the transfer map. The second bound $d'\leq td$ follows from the fact that a lifted logical operator has weight $td$. If $k'=k$, the transfer and projection maps on the logical spaces are isomorphisms. Therefore, a minimum-weight logical operator of the lifted code projects to a nontrivial logical operator of the base code without increasing its weight, which implies $d'\geq d$.
\end{proof}

\section{Search for cover codes}\label{Sec:CodeSearch}
In this section, we provide finite-size instances of LP codes obtained by our framework. For each group extension, we construct the corresponding lifted code, evaluate its parameters, and retain the code that performs best according to a chosen search criterion. The search criteria can be given as input to the algorithm, and we use $kd^2/n$ here. The number of logical qubits can be readily calculated by Gaussian elimination. To search for low-weight logical operators, we use \texttt{QDistRnd}~\cite{Pryadko_2022} with $2,000$ sampling iterations. The smallest logical-operator weight found provides an upper bound on the code distance and is used as its estimate during the search. The code search is performed using \texttt{GAP}~\cite{GAP4}. Further details, including pseudocode, are provided in Appendix~\ref{Appendix:DetailsofNumerics_codesearch}. After the codes are found, we verify the distance exactly using mixed-integer programming~\cite{Bravyi2024, PhysRevResearch.6.033086} in some cases. The mixed-integer programming is performed using \texttt{Python}. In reporting the code parameters, we write $\leq d$ when only an upper bound has been obtained and write $d$ when the distance has been verified exactly.

\vspace{5pt} \noindent {\bf Lift of the BB code.} As a benchmark, we lift BB codes studied in Ref.~\cite{symons2025sequences}. We first consider the lift of a base BB code defined for the group $\Z_6\times \Z_6$, giving the following chain complex over the ring $R = \F_2[\Z_6\times \Z_6]\simeq \F_2[x, y]/(x^6 - 1, y^6 - 1)$
\begin{align}
    R \xlongrightarrow{\begin{pmatrix}
        a\\
        b
    \end{pmatrix}} R^2 \xlongrightarrow{\begin{pmatrix}
        b & a
    \end{pmatrix}} R.
\end{align}
We choose the boundary maps to be $a = x^3 + y + y^2,\; b = y^3 + x + x^2$, whose code parameter is $\llbracket72, 12, 6\rrbracket$.

Table~\ref{table:Lifted_BBcodes} shows the lifted LP codes (see the rows colored red). We extend the base group $G=\Z_6\times\Z_6$ by the kernel $K=\Z_3$ through the short exact sequence in Eq.~\eqref{Eq:group_extension}, obtaining several possible extension groups $H$. The first two columns list the kernel $K$ and the extended group $H$ respectively. The last column shows the code parameters of the resulting codes. The code with the best $kd^2/n$ for each extension is presented. We compare these codes with the lifted BB codes reported in Ref.~\cite{symons2025sequences}, with improvements indicated in bold. For example, while Ref.~\cite{symons2025sequences} reports a code with parameters $\llbracket216,12,12\rrbracket$, our search finds codes with the same blocklength and number of logical qubits but an improved distance of $14$.

We also consider lifts of a BB code defined by the group $\Z_8\times \Z_4$, with the boundary maps $a = xy^3 + 1 + x^6 + x^3y^2$, $b = x^6y + x^4y + x^3 + x^5y$. The code parameter of the base code is $\llbracket64, 14, 8\rrbracket$. The corresponding lifted codes are shown in the blue rows of Table~\ref{table:Lifted_BBcodes}. We find a code with the parameter $\llbracket128,16,12\rrbracket$. This improves upon the best instance $\llbracket128,14,12\rrbracket$ reported in Ref.~\cite{symons2025sequences}.

\begin{table}[t]
\centering
\setlength{\tabcolsep}{0pt}      % default is about 6pt
\renewcommand{\arraystretch}{1.15}

\begin{tabular}{|c|c|c|}
\hline
\;\;Group $K$\;\; & Group $H$ & $\llbracket n, k, d\rrbracket$ \\
\hline

\rowcolor{red!15}
$\Z_3$
&
\makecell[c]{
$\Z_{18}\times\Z_6$ (in Ref.~\cite{symons2025sequences})\\
$\,\,\,\,\Z_2\times\Z_2\times ((\Z_3\times\Z_3)\rtimes\Z_3)\,\,\,\,$\\
$\Z_6\times\Z_6\times\Z_3$
}
&
\makecell[c]{
$\,\,\,\llbracket 216,12,\textbf{12}\rrbracket\,\,\,$\\
$\llbracket 216,12,\textbf{14}\rrbracket$\\
$\,\,\llbracket 216,12,\textbf{14}\rrbracket\,\,$
}
\\
\hline

\rowcolor{blue!10}
$\Z_2$
&
\makecell[c]{\quad\quad
$\Z_8\times\Z_8$ (in Ref.~\cite{symons2025sequences})\quad\;\;\;\\
$\Z_{16}\rtimes\Z_4$
}
&
\makecell[c]{
$\,\,\,\llbracket 128,\textbf{14},12\rrbracket\,\,\,$\\
$\,\,\,\llbracket 128,\textbf{16},12\rrbracket\,\,\,$
}
\\
\hline

\end{tabular}
\caption{\textbf{Lifted BB codes.}
We compare the lifted BB codes in Ref.~\cite{symons2025sequences}, and the ones constructed using our method.
For the first block, the base code is defined for the group $\mathbb{Z}_6\times \mathbb{Z}_6$ and the parity-check polynomials 
$a = x^3 + y + y^2$ and $b = y^3 + x + x^2$. 
The code parameter of the base code is $\llbracket 72, 12, 6\rrbracket$. 
The lifted codes are shown in the first row block of the table, colored in red. 
For the second block, the base code is defined for the group $\mathbb{Z}_8\times \mathbb{Z}_4$, with the polynomials 
$a = xy^3 + 1 + x^6 + x^3y^2$ and 
$b = x^6y + x^4y + x^3 + x^5y$. 
The base code parameter is $\llbracket64, 14, 8\rrbracket$. 
The second row block shows the lifted codes, colored in blue.
}
\label{table:Lifted_BBcodes}
\end{table}

\vspace{5pt} \noindent {\bf Lift of the CC code.} We further apply our lifting framework to the clustered cyclic (CC) codes introduced in Ref.~\cite{gu2026qgpuparallellogicquantum}. CC codes are a subfamily of LP codes defined over the group algebra $R = \F_2[\Z_p]$ with $p$ being prime. They are constructed from the boundary maps in Eq.~\eqref{Eq:LPClassicalCodes}, subject to the following conditions~ $(i)$ every entry of $\partial_A$ and $\partial_B$ is either zero or a binomial; $(ii)$ the zero entries are arranged cyclically; and $(iii)$ both matrices are square and have full rank over $R$. These properties are important for properties of the CC codes and their lifts in Sec.~\ref{Subsec:PPS}.

For instance, we consider a base CC code defined by the group $G = \Z_3$, with the boundary maps
\begin{align}
    \partial_A = \begin{pmatrix}
        1 + x^2 & x + x^2\\
        1 + x & x + x^2
    \end{pmatrix}, \quad 
    \partial_B = \begin{pmatrix}
    1 + x^2 & x + x^2\\
    1 + x & 1 + x
    \end{pmatrix}.
\end{align}
This code has the parameter $\llbracket24, 8, 3\rrbracket$~\cite{gu2026qgpuparallellogicquantum}. We search for lifts with kernels $\Z_3$ and $\Z_5$, using the code distance as the search criterion. The resulting lifted codes have parameters $\llbracket72,8,8\rrbracket$ and $\llbracket120,8,13\rrbracket$, respectively.

We also lift several other CC codes. We restrict to lifts with kernels $K$ of odd order, and search for codes using distance as the optimization criterion. The results are summarized in Table~\ref{Table:liftedCCcodes}. The left two columns show the group and the parameter of the base code respectively, and the right two columns show the group and the best code parameter of the lifted codes. The boundary maps of the lifted codes are presented in Appendix~\ref{App:boundarymaps}.

As shown in Table~\ref{Table:liftedCCcodes}, we identify several finite-size lifted CC codes with improved distances. Some of these codes also improve upon the parameters reported in Ref.~\cite{gu2026qgpuparallellogicquantum}; they achieve the same $k$ and $d$ using fewer physical qubits. For instance, the lifted CC code $\llbracket120, 8, 14\rrbracket$ improves upon the CC code $\llbracket 136, 8, \leq 14\rrbracket$. Moreover, as we discuss in Sec.~\ref{Sec:LogicalGadget}, these lifted codes admit the same parallel product surgery~\cite{gu2026qgpuparallellogicquantum} as the base CC codes. They therefore preserve the logical-operation capabilities with an improved code distance.

\begin{table}
\centering

\setlength{\tabcolsep}{6pt}      % default is about 6pt
\renewcommand{\arraystretch}{1.15}

\begin{tabular}{c|c|c|c}
\toprule
Group $G$ & Base $\llbracket n,k,d\rrbracket $ & Group $H$ & Lifted $\llbracket n,k,d\rrbracket $\; \\
\midrule
$\Z_3$ & $\llbracket 24, 8, 3\rrbracket$ &$\Z_9$ & $\llbracket 72, 8, 8\rrbracket$\\
$\Z_3$ & $\llbracket 24, 8, 3\rrbracket$ &$\Z_{15}$ & $\llbracket120, 8, 13\rrbracket$\\
$\Z_5$ & $\llbracket 40, 8, 5\rrbracket$ &$\Z_{15}$ & $\llbracket120, 8, 14\rrbracket$\\
$\Z_7$ & $\llbracket 56, 8, 7\rrbracket$ &$\Z_{21}$ & $\llbracket168, 8, 15\rrbracket$\\
$\Z_3$ & $\llbracket 54, 18, 3\rrbracket$ &$\Z_{9}$ & $\llbracket 162, 18, 8\rrbracket$\\
$\Z_{11}$ & $\llbracket 88, 8, 10\rrbracket$ & - & -\\
$\Z_{13}$ & $\llbracket 104, 8, 11\rrbracket$ & - & -\\
$\Z_{17}$ & $\llbracket 136, 8, \leq 14\rrbracket$ & - & -\\
\bottomrule
\end{tabular}
\caption{
\textbf{Lifted CC codes.} We extend the cyclic group $G = \Z_l$ to the group $H$ listed in the table. The left two columns show the group and the code parameters of the base code, and the right two columns show the group and the code parameters of the lifted code, respectively. The boundary maps are given in Appendix~\ref{App:boundarymaps}. The last three rows list previously reported CC codes for comparison. 
}
\label{Table:liftedCCcodes}
\end{table}

\section{Logical operation gadgets}\label{Sec:LogicalGadget}

The lifting framework provides a natural setting for transferring logical-operation gadgets from a base LP code to its lifts.
In Sec.~\ref{Subsec:LiftedSurgery}, we compose the surgery maps of the projected code with the corresponding transfer maps to construct homological measurements~\cite{PhysRevX.15.021088} on the lifted LP code. We demonstrate that the ancilla-code space overhead for surgery can be reduced through this construction.
In Sec.~\ref{Subsec:PPS}, we turn to the lifted CC codes obtained in Sec.~\ref{Sec:CodeSearch} and establish conditions under which they retain the parallel product surgery~\cite{gu2026qgpuparallellogicquantum}.

\subsection{Code surgery}\label{Subsec:LiftedSurgery}

\noindent {\bf Review.} We begin by reviewing the mapping-cone formulation of code surgery~\cite{PhysRevX.15.021088}. Let $Q$ be the data code described by the chain complex ($Q_2$ being $X$-checks)
\begin{equation}
    \begin{tikzcd}
    {{Q}_2} & {{Q}_1} & {Q_0} 
    \arrow["{\partial_2}", from=1-1, to=1-2]
    \arrow["{{\partial}_1}", from=1-2, to=1-3]
    \end{tikzcd}.
\end{equation}
To measure logical operators of $Q$, we introduce an ancilla code $A$ described by ($A_1$ being $X$-checks)
\begin{equation}
    \begin{tikzcd}
    {A_1} & {A_0} & {A_{-1}}
    \arrow["{\partial_1^A}", from=1-1, to=1-2]
    \arrow["{\partial_0^A}", from=1-2, to=1-3]
    \end{tikzcd},
\end{equation}
where each qubit in $A_0$ is prepared in the physical $\ket{0}$ state. Then, we merge the ancilla code $A$ with the data code $Q$ using the maps $\Gamma_{0,1}$ such that the following diagram commutes, i.e., $\partial_1 \circ \Gamma_1 = \Gamma_0 \circ \partial_1^A$
\begin{equation}\label{eq:surgery_diagram_base}
    \begin{tikzcd}
    {A_1} & {A_0} & {A_{-1}} \\
    {Q_2} & {Q_1} & {Q_0} 
    \arrow["{\partial_1^A}", from=1-1, to=1-2]
    \arrow["{\partial_0^A}", from=1-2, to=1-3]
        \arrow["{{\partial}_2}", from=2-1, to=2-2]
    \arrow["{{\partial}_1}", from=2-2, to=2-3]
        \arrow["{\Gamma_1}", from=1-1, to=2-2]
    \arrow["\Gamma_0", from=1-2, to=2-3]
    \end{tikzcd}.
\end{equation}
The mapping cone of this chain map defines a CSS code, which we refer to as the \emph{merged code} and denote by $\operatorname{cone} (\Gamma)$. It has been shown that the phenomenological distance of the code surgery is given by the dressed distance of data logical operators in the merged code~\cite{Williamson2026, zheng2025highratesurgeryconstantoverheadlogical}. Here the data logical operators mean the logical operators in $\operatorname{cone}(\Gamma)$ originating from $Q$, and we treat those originating from $A$ as gauge logical operators. See Appendix~\ref{App:surgery} for the details.

The logical measurement is implemented by measuring the newly introduced $X$-type checks indexed by $A_1$ in the merged code. The ancillary qubits in $A_0$ are initialized in the $Z$-basis $|0 \rangle$ state, so any individual $X$-check acting nontrivially on $A_0$ has a random outcome. Therefore, a product of merged $X$-checks defines a well-defined observable on $Q$ only when its ancillary support cancels.
Such products are labeled by $a \in \ker \partial_1^A$ and measure the product of $X$ specified by $\Gamma_1 a$. The commutativity condition implies $ \partial_1\Gamma_1a
= \Gamma_0\partial_1^Aa = 0$, so $\Gamma_1a$ represents either an $X$ stabilizer or a logical $X$ operator of $Q$:
\begin{align} \label{eq:surgery_measurement}
    \Gamma_1(\ker\partial_1^A)\subseteq \ker \partial_1 = L_X\oplus S_X.
\end{align}
Accordingly, the surgery data $(\Gamma_1, \partial_1^A)$ should be chosen so that $\Gamma_1(\ker\partial_1^A)$ contains the desired logical operators, while the merged code defined by Eq.~\eqref{eq:surgery_diagram_base} retains the fault-tolerance properties required throughout the measurement. 
In many code-surgery constructions, the required ancilla size is controlled by the weight of the logical operators to be measured~\cite{doi:10.1126/sciadv.abn1717, cross2025improvedqldpcsurgerylogical, Williamson2026, zheng2025highratesurgeryconstantoverheadlogical}.

\vspace{5pt} \noindent {\bf Proposal.} Here we propose a method to reduce the ancilla overhead in the surgery of LP codes, which we call \emph{lifted code surgery}. The idea is to first construct an ancilla code $A$ and chain maps $\Gamma_1$ and $\Gamma_0$ for the projected LP code $Q_\pi$, and then use them to perform surgery on an index-$t$ lift $Q$ of $Q_\pi$. If a logical measurement on $Q_\pi$ is implemented using $A$ and the chain maps $\Gamma_1$ and $\Gamma_0$, then the chain maps from $A$ to $Q$ are obtained by composing them with the transfer maps. A key advantage is that the ancilla code $A$ can require fewer qubits than an ancilla constructed directly for $Q$, as the projected code $Q_{\pi}$ often has lower-weight representatives of logical operators than those of $Q$. We demonstrate this reduction below using the gross code and other BB codes.

Suppose that a code-surgery construction for the projected code $Q_{\pi}$ is specified by an ancilla code $A$ and maps
\begin{align}
    \Gamma_1 &: A_1 \longrightarrow Q_1^{\pi}, \\
    \Gamma_0 &: A_0 \longrightarrow Q_0^{\pi}.
\end{align}
The corresponding \emph{lifted code surgery} on $Q$ uses the same ancilla code $A$, with chain maps
\begin{align}
    \widetilde{\Gamma}_1
    &= \tau_1\circ\Gamma_1, \\
    \widetilde{\Gamma}_0
    &= \tau_0\circ\Gamma_0,
\end{align}
where 
\begin{align}
    \tau_i:Q_i^\pi\longrightarrow Q_i
\end{align}
denote the transfer maps.
Equivalently, the construction is described by the following commutative diagram:
\begin{equation}\label{eq:surgery_diagram_cover}
    \begin{tikzcd}
        {A_1} & {A_0} & {A_{-1}} \\
        {Q^\pi_2} & {Q^\pi_1} & {Q^\pi_0} \\
        {Q_2} & {Q_1} & {Q_0}
        \arrow["{\partial_1^A}", from=1-1, to=1-2]
        \arrow["{\partial_0^A}", from=1-2, to=1-3]
        \arrow["{\partial_2^\pi}", from=2-1, to=2-2]
        \arrow["{\partial_1^\pi}", from=2-2, to=2-3]
        \arrow["{\partial_2}", from=3-1, to=3-2]
        \arrow["{\partial_1}", from=3-2, to=3-3]
        \arrow["{\Gamma_1}", from=1-1, to=2-2]
        \arrow["{\Gamma_0}", from=1-2, to=2-3]
        \arrow["{\tau_2}", from=2-1, to=3-1]
        \arrow["{\tau_1}", from=2-2, to=3-2]
        \arrow["{\tau_0}", from=2-3, to=3-3]
    \end{tikzcd}.
\end{equation}
Thus, lifted code surgery uses the ancilla $A$ designed for $Q_{\pi}$ and couples it to the lifted code $Q$ through the composite maps $\tau_1\circ\Gamma_1$ and $\tau_0\circ\Gamma_0$.

We refer to the construction in Eq.~\eqref{eq:surgery_diagram_cover} as \emph{lifted code surgery}. The measured logical operators are represented by
\begin{align}\label{Eq:MeasuredLogicalLiftedSurgery}
\tau_1\circ\Gamma_1\left(\ker\partial_1^A\right)
\subseteq \ker\partial_1.
\end{align}
At the level of homology, these operators are obtained by applying the induced transfer map $\hat{\tau}_1$ to the logical operators measured by the corresponding surgery on $Q_{\pi}$. When $t$ is odd, $\hat{\tau}_1$ is injective, so every nontrivial logical operator measured on $Q_{\pi}$ is mapped to a nontrivial logical operator of $Q$. However, $\hat{\tau}_1$ need not be surjective, so lifted code surgery only addresses logical operators lying in the image of $\hat{\tau}_1$ and may not provide access to all logical operators of $Q$. When $t$ is even, $\hat{\tau}_1$ need not be injective, so some of the logical operators measured on $Q_{\pi}$ may be mapped to trivial homology classes in $Q$. In practice, one must therefore explicitly analyze $\hat{\tau}_1$ and choose the ancilla code $A$ and the map $\Gamma_1$ so that the desired nontrivial logical operators of $Q$ are measured.

\vspace{5pt} \noindent {\bf Caveats.} We emphasize two points. First, the coupling between the ancilla code $A$ and the lifted code $Q$ is generally one-to-$t$, rather than one-to-one, because the chain maps are given by the composites $\tau_1\circ\Gamma_1$ and $\tau_0\circ\Gamma_0$. Consequently, relative to their weights in the ancilla code, the weights of the $X$-type checks associated with $A_1$ increase by $t$, rather than by one as in the one-to-one coupling case. Also, the degrees of the ancilla qubits $A_0$ with respect to the $Z$-type checks increase by $t$.

Second, preservation of the merged-code distance for $Q_{\pi}$ does not imply preservation of the merged-code distance for $Q$. Accordingly, we use lifted code surgery only as ancilla initialization, and then we require that $(i)$ the target logical operators are represented by the space $\tau_1\circ\Gamma_1(\ker\partial_1^A)$ and $(ii)$ the maps $\widetilde{\Gamma}_1$ and $\widetilde{\Gamma}_0$ satisfy the chain-map compatibility conditions required for a valid surgery construction. Starting from this initialization, we randomly add qubits and checks to the ancilla code until the desired merged-code distance for $Q$ is preserved. Although there is no theoretical guarantee that the resulting boosted ancilla code has lower space overhead than one constructed directly for $Q$, we numerically demonstrate the reduction in the space overhead through several examples. Further details of the ancilla code construction are provided in Appendix~\ref{App:surgery}.

\vspace{5pt} \noindent {\bf Demonstration.}
First, we study lifted code surgery of a low-rate type for the gross code $\llbracket 144,12,12\rrbracket$, which is based on auxiliary graph surgery~\cite{Williamson2026} (simply graph surgery). The projected codes of the gross code are listed in Table~\ref{Table:GrossProjectedCodes}, and we choose the projected code $\llbracket 72,12,6\rrbracket$ among them. For the logical basis defined in Appendix~\ref{App:GrossLogical}, the transferred logical subspace $\hat{\tau}_1\left(H_1(Q_{\pi})\right)$ is spanned by $\overline{X}_1,\overline{X}_2,\ldots,\overline{X}_6$. 

We consider measuring each of these six logical operators, and compare the ancilla-code sizes required by lifted graph surgery with those required by the standard graph-surgery construction of Ref.~\cite{Williamson2026}. We refer to the latter as \emph{direct graph surgery}. In each case, we first initialize the ancilla graph $A$ and the chain maps $\Gamma_0$ and $\Gamma_1$ so that the target logical operator is measured. We then add edges to the ancilla graph to augment the merged-code distance until no logical operator of weight below $12$ is found in $10{,}000$ \texttt{QDistRnd} samples.

The results are summarized in Table~\ref{tab:q41-wy-overhead}. The first column specifies the measured logical operator, and the second identifies the construction stage; i.e., immediately after initialization and after the estimated distance is preserved. The third and fourth columns show the ancilla-code sizes for lifted and direct graph surgery, respectively. Each tuple lists the numbers of ancilla qubits, $X$-checks, and $Z$-checks, in this order. Table~\ref{tab:q41-wy-x1-distributions} shows the distributions of check weights and qubit degrees of the merged code for the measurement of $\overline{X}_1$. Table~\ref{tab:q41-wy-x1-distributions}(a) corresponds to the lifted surgery, and Table~\ref{tab:q41-wy-x1-distributions}(b) shows those of the direct surgery. The lifted merged codes for the other logical operators exhibit similar distributions. Overall, lifted graph surgery reduces the ancilla overhead by approximately a factor of two, while introducing only a modest increase in check weights and qubit degrees. Note that the check-weight of the data code is $6$, and the qubit-degree is $4$. 

\begin{table}[t] 
\centering 
\setlength{\tabcolsep}{6pt} % default is about 6pt 
\renewcommand{\arraystretch}{1.15} 
    \begin{tabular}{c|c|c|c}
        \hline
        Target & Stage & Lifted overhead & Direct overhead \\
        \hline
        $\overline{X}_1$
        & \makecell[c]{Init. \\ Final}
        & \makecell[c]{
            $(9, 6, 4) = 19$ \\
            $(12, 6, 7) = 25$
          }
        & \makecell[c]{
            $(18, 12, 7) = 37$ \\
            $(25, 12, 14) = 51$
          } \\
        \hline
        $\overline{X}_2$
        & \makecell[c]{Init. \\ Final}
        & \makecell[c]{
            $(9, 6, 4) = 19$ \\
            $(12, 6, 7) = 25$
          }
        & \makecell[c]{
            $(18, 12, 7) = 37$ \\
            $(24, 12, 13) = 49$
          } \\
        \hline
        $\overline{X}_3$
        & \makecell[c]{Init. \\ Final}
        & \makecell[c]{
            $(9, 6, 4) = 19$ \\
            $(12, 6, 7) = 25$
          }
        & \makecell[c]{
            $(18, 12, 7) = 37$ \\
            $(24, 12, 13) = 49$
          } \\
        \hline
        $\overline{X}_4$
        & \makecell[c]{Init. \\ Final}
        & \makecell[c]{
            $(9, 6, 4) = 19$ \\
            $(12, 6, 7) = 25$
          }
        & \makecell[c]{
            $(18, 12, 7) = 37$ \\
            $(22, 12, 11) = 45$
          } \\
        \hline
        $\overline{X}_5$
        & \makecell[c]{Init. \\ Final}
        & \makecell[c]{
            $(9, 6, 4) = 19$ \\
            $(12, 6, 7) = 25$
          }
        & \makecell[c]{
            $(18, 12, 7) = 37$ \\
            $(25, 12, 14) = 51$
          } \\
        \hline
        $\overline{X}_6$
        & \makecell[c]{Init. \\ Final}
        & \makecell[c]{
            $(9, 6, 4) = 19$ \\
            $(12, 6, 7) = 25$
          }
        & \makecell[c]{
            $(18, 12, 7) = 37$ \\
            $(23, 12, 12) = 47$
          } \\
        \hline
    \end{tabular}
    \caption{\textbf{Ancilla overhead for the surgery protocols of the gross code.} For each target logical operator, the upper and lower entries correspond to the initialization and final stages, respectively. ``Init.'' denotes the merged code immediately after initialization, while ``Final'' denotes the merged code after the estimated distance is restored. Each entry is reported as $(n_A,r_X,r_Z)=n_A+r_X+r_Z$, where $n_A$ is the number of ancilla qubits and $r_X$ and $r_Z$ are the numbers of ancilla $X$ and $Z$ checks, respectively.}
\label{tab:q41-wy-overhead}
\end{table}

\begin{table}[t]
\centering
\setlength{\tabcolsep}{4pt}
\renewcommand{\arraystretch}{1.15}
\resizebox{0.98\columnwidth}{!}{%
\begin{tabular}{@{}cc|cc|cc|cc|cc|cc@{}}
    \multicolumn{6}{c|}{\textbf{(a) Lifted surgery}}
    &
    \multicolumn{6}{c}{\textbf{(b) Direct surgery}}
    \\[3pt] 
    \cline{1-6}\cline{7-12}

    \multicolumn{2}{c|}{$X$ stab.}
    & \multicolumn{2}{c|}{$Z$ stab.}
    & \multicolumn{2}{c|}{Qubits}
    &
    \multicolumn{2}{c|}{$X$ stab.}
    & \multicolumn{2}{c|}{$Z$ stab.}
    & \multicolumn{2}{c}{Qubits}
    \\

    {\bf wt.} & \#
    & {\bf wt.} & \#
    & {\bf deg.} & \#
    &
    {\bf wt.} & \#
    & {\bf wt.} & \#
    & {\bf deg.} & \#
    \\
    \hline

    4 & 0  & 3 & 6  & 3 & 2
    & 4 & 2  & 3 & 8  & 3 & 2   \\

    5 & 1  & 4 & 1  & 4 & 0
    & 5 & 7  & 4 & 5  & 4 & 10  \\

    6 & 76 & 5 & 0  & 5 & 5
    & 6 & 74 & 5 & 1  & 5 & 7   \\

    7 & 1  & 6 & 54 & 6 & 134
    & 7 & 1  & 6 & 54 & 6 & 138 \\

      &    & 7 & 18 & 7 & 15
    &   &    & 7 & 18 & 7 & 12  \\
    \hline

    Tot. & 78 & Tot. & 79 & Tot. & 156
    & Tot. & 84 & Tot. & 86 & Tot. & 169 \\
    \hline
\end{tabular}%
}
\caption{\textbf{Detailed analytics of the surgery for $\overline{X}_1$.} We compare the final merged codes obtained using {\bf (a)} lifted graph surgery and {\bf (b)} direct graph surgery. }
\label{tab:q41-wy-x1-distributions}
\end{table}

% \begin{table*}[!t]
% \centering
% \setlength{\tabcolsep}{5pt}
% \renewcommand{\arraystretch}{1.15}

% \begin{minipage}[t]{0.48\textwidth}
% \centering
% \textbf{(a) Lifted surgery}\\[3pt]
% \resizebox{\linewidth}{!}{%
% \begin{tabular}{cc|cc|cc}
%     \hline
%     \multicolumn{2}{c|}{$X$ checks}
%     & \multicolumn{2}{c|}{$Z$ checks}
%     & \multicolumn{2}{c}{Qubits} \\
%     Weight & Count & Weight & Count & Degree & Count \\
%     \hline
%      4 & 0 & 3 & 6  & 3 & 2   \\
%     5 & 1  & 4 & 1  &  4 & 0  \\
%     6 & 76  & 5 & 0 & 5 & 5 \\
%     7 & 1 & 6 & 54 & 6 & 134  \\
%     & & 7 & 18 & 7 & 15 \\
%     \hline
%     Total  & 78 & Total & 79 & Total & 156\\
%     \hline
% \end{tabular}%
% }
% \end{minipage}
% \hfill
% \begin{minipage}[t]{0.48\textwidth}
% \centering
% \textbf{(b) Direct surgery}\\[3pt]
% \resizebox{\linewidth}{!}{%
% \begin{tabular}{cc|cc|cc}
%     \hline
%     \multicolumn{2}{c|}{$X$ checks}
%     & \multicolumn{2}{c|}{$Z$ checks}
%     & \multicolumn{2}{c}{Qubits} \\
%     Weight & Count & Weight & Count & Degree & Count \\
%     \hline
%     4 & 2  & 3 & 8  & 3 & 2   \\
%     5 & 7  & 4 & 5  & 4 & 10  \\
%     6 & 74 & 5 & 1  & 5 & 7   \\
%     7 & 1  & 6 & 54 & 6 & 138 \\
%       &    & 7 & 18 & 7 & 12  \\
%     \hline
%     Total  & 84 & Total & 86 & Total & 169\\
%     \hline
% \end{tabular}%
% }
% \end{minipage}

% \caption{\textbf{Check-weight and qubit-degree distributions for
% $\overline{X}_1$.} The final merged codes obtained using (a) lifted graph surgery and (b) direct graph surgery are compared. }
% \label{tab:q41-wy-x1-distributions}
% \end{table*}

Second, we study lifted code surgery of a high-rate type. Here the ancilla codes are constructed using the randomized procedure of Ref.~\cite{zheng2025highratesurgeryconstantoverheadlogical}, with further details provided in Appendix~\ref{App:HighRateSurgery}. As the data code, we choose the lifted BB code $\llbracket216,12,14\rrbracket$, listed in the last row of the red block in Table~\ref{table:Lifted_BBcodes}. We take the base BB code $\llbracket72,12,6\rrbracket$ as the projected code $Q_{\pi}$. For this pair of codes, the induced transfer map $\hat{\tau}_1$ on homology is bijective. Consequently, every logical operator of $Q$ lies in the image of $\hat{\tau}_1$ and can be addressed by the lifted surgery construction.

We compare the ancilla code overhead required to simultaneously measure multiple logical $X$ operators using direct and lifted high-rate surgery. For each $m=1,2,\ldots,12$, we randomly sample five sets of $m$ logical $X$ operators and apply both constructions to the same sampled sets. Throughout the construction, both the maximum check weight and the maximum qubit degree are constrained to be at most $9$. For each merged code, we perform $10{,}000$ randomized \texttt{QDistRnd}-style samples and add ancilla qubits and checks until no dressed data logical operator of weight below $14$ is found. We then perform an additional $30{,}000$ samples and accept the code if no such logical operator is found.

The normalized space cost, defined as the total number of qubits and $X$- and $Z$-type checks in the merged code divided by that of the data code, is shown in Fig.~\ref{Fig:SurgeryOverhead}, with blue points representing direct surgery and red points representing lifted surgery. The dashed lines describe the space overhead after initialization, and the solid lines show the overhead after the merged code passed the randomized distance-screening criterion. The results show that lifted surgery can substantially reduce the space overhead compared to the direct surgery construction. Also, the lifted surgery requires fewer qubits and checks to restore the estimated distance.

\begin{figure}[t]
    \centering
    \includegraphics[width=0.97\linewidth]{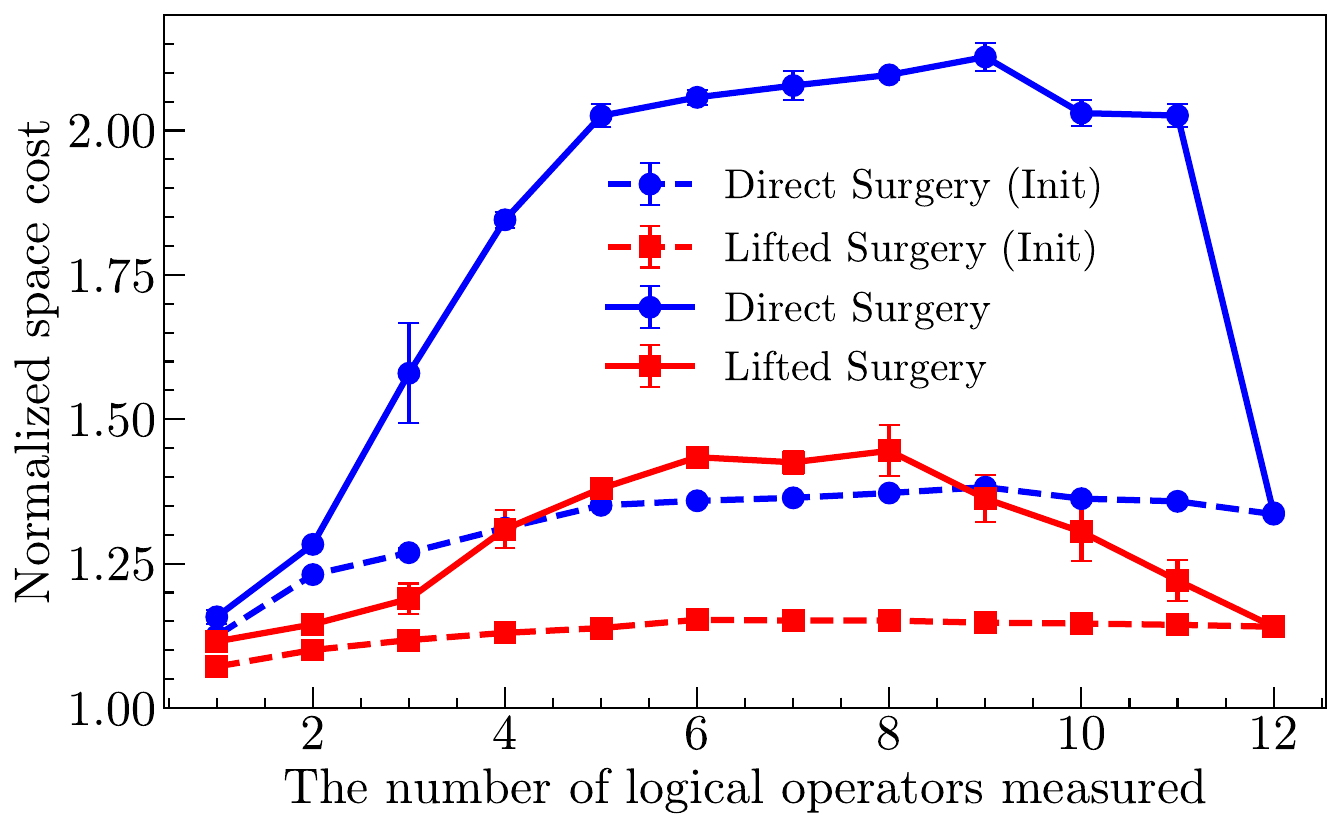}
    \caption{\textbf{Normalized space cost of direct code surgery and the lifted code surgery.} The horizontal axis shows the number of logical operators measured, and the vertical axis shows the normalized space cost, which is the total number of qubits and checks in the merged code divided by that of the data code. We randomly select five sets of $m$ logical operators for each $m = 1, \ldots, 12$, whose normalized space cost fluctuates across the sets as shown using the error bars. The dashed lines show the space overhead after initialization, whereas solid lines indicate the overhead after the estimated distance has been restored.}
    \label{Fig:SurgeryOverhead}
\end{figure}

\noindent {\bf Remark.} Although we phrased the construction as lifted surgery, an analogous construction can be applied to general qLDPC codes without explicitly invoking chain maps. In the gauging picture~\cite{Williamson2026}, lifted graph surgery groups the data qubits into blocks of size $t$ and introduces gauge qubits between these blocks. More generally, one may partition the physical qubits supporting a logical operator into arbitrary subsets, treat these subsets as vertices of an ancilla graph, and place gauge qubits on its edges. The adjacent $Z$ checks are then dressed to commute with the newly introduced local symmetry generators. Gauging these symmetries yields an ancilla construction that measures the target logical operator. Hence, the space-saving mechanism underlying lifted graph surgery could be more universal than our treatment, although there remains some subtlety of distance preservation.

On the other hand, the situation is less straightforward for high-rate surgery~\cite{zheng2025highratesurgeryconstantoverheadlogical}. In that setting, constructing a suitable ancilla code together with nontransversal chain maps seems difficult without the algebraic structure supplied by the lift and its transfer maps. Developing systematic nontransversal surgery constructions for general qLDPC codes, without assuming an underlying lifted structure, is therefore an interesting direction for future work.

\subsection{Parallel product surgery}\label{Subsec:PPS}
We study parallel product surgery~\cite{gu2026qgpuparallellogicquantum} (PPS) for the \emph{lifted} CC codes introduced in Sec.~\ref{Sec:CodeSearch}. CC codes were originally proposed as a subfamily of LP codes with a clustered logical basis. The structure of CC codes allows the logical operators addressed by PPS to be characterized explicitly. 
In general, lifted CC codes need not satisfy the defining properties of the CC codes, so the logical addressability does not automatically extend. We identify sufficient conditions under which a lifted CC code retains the same logical addressability under PPS as its base CC code, and show that all lifted CC codes listed in Table~\ref{Table:liftedCCcodes} satisfy these conditions.

\begin{definition}[Parallel product surgery~\cite{gu2026qgpuparallellogicquantum}]
    Suppose that an LP code is defined by the pair of classical codes $A_\bullet$ and $B_\bullet$ in Eq.~\eqref{Eq:LPClassicalCodes}. Let $H_a: A_1\to A_0$ and $H_b: B_1\to B_0$ be $R$-linear maps. We take a copy of the data LP code as the ancilla code and use the following chain maps
    \begin{align}
        \Gamma_1 &= I\otimes H_b\oplus H_a\otimes I,\\
        \Gamma_0 &= \big[ H_a\otimes I, I\otimes H_b \big],
    \end{align}
    for code surgery in Eq.~\eqref{eq:surgery_diagram_base}.      This is called \emph{parallel product surgery}.

    % Equivalently, this is described by the following diagram:
    % \begin{eqnarray}
    %     \begin{tikzcd}
    %        {} & {A_1\otimes B_0} & {} \\
    %        {A_1\otimes B_1} && {A_0\otimes B_0} \\
    %        {} & {A_0\otimes B_1} & \\
    %        {} & {A_1\otimes B_0} & {} \\
    %        {A_1\otimes B_1} && {A_0\otimes B_0} \\
    %        {} & {A_0\otimes B_1} &
    %        \arrow["{I\otimes\partial_B}", from=2-1, to=1-2]
    %        \arrow["{\partial_A\otimes I}", from=1-2, to=2-3]
    %        \arrow["", from=2-1, to=3-2]
    %        \arrow["", from=3-2, to=2-3]
    %        \arrow["", from=5-1, to=4-2]
    %        \arrow["", from=4-2, to=5-3]
    %        \arrow["{\partial_A\otimes I}"', from=5-1, to=6-2]
    %        \arrow["{I\otimes\partial_B}"', from=6-2, to=5-3]
    %        \arrow[
    %            "{\textcolor{red}{I\otimes H_b}}",
    %            red,
    %            from=2-1,
    %            to=4-2,
    %            sloped]
    %        \arrow[
    %            "{\textcolor{red}{I\otimes H_b}}"',
    %            red,
    %            from=3-2,
    %            to=5-3,
    %            sloped]
    %        \arrow[
    %            "{\textcolor{red}{H_a\otimes I}}",
    %            red,
    %            from=2-1,
    %            to=6-2,
    %            sloped]
    %        \arrow[
    %            "{\textcolor{red}{H_a\otimes I}}",
    %            red,
    %            from=1-2,
    %            to=5-3,
    %            sloped]
    %     \end{tikzcd},
    % \end{eqnarray}
    % where the red arrows correspond to the chain maps $\Gamma_0$ and $\Gamma_1$.
\end{definition}

PPS is defined for general LP codes, but for arbitrary choices of $H_a$ and $H_b$, the measured logical operators (cf. Eq.~\eqref{eq:surgery_measurement}) are not generally available in an explicit form. For CC codes, however, the measured logical subspace can be characterized completely for any choice of $H_a$ and $H_b$; the defining conditions of CC codes imply~\cite{gu2026qgpuparallellogicquantum}
\begin{align}\label{Eq:KernelForPPS}
    \ker \partial_2 = \ker \partial_A\otimes_R \ker \partial_B = \mathsf{row} (\omega_G I_a\otimes_R I_b),
\end{align}
where $\mathsf{row}(\cdot)$ denotes the space spanned by the row vectors of a matrix, and $\omega_G = \sum_{g\in G}g$. The logical $X$ operators measured by PPS are represented by
\begin{align}\label{Eq:PPSforCCcodes}
    \Gamma_1(\ker \partial_2) = \mathsf{row}\left[\omega_G (  I\otimes_R H_b  \oplus H_{a}\otimes_R I )\right]
\end{align}
Moreover, because $\omega_G g = g\omega_G=\omega_G$ for every $g\in G$, multiplication by $\omega_G$ maps each entry of $H_a$ and $H_b$ to either $0$ or $\omega_G$, depending only on the parity of the entries. Therefore, without loss of generality, each entry of $H_a$ and $H_b$ may be taken to be either $0$ or $1$ when characterizing the measured logical subspace.

Also, it has been shown that the logical operators of the CC codes admit the canonical representation~\cite{gu2026qgpuparallellogicquantum}
\begin{align}\label{Eq:LogicalDistribution}
    L_{X/Z} &= \operatorname{diag}_{2n_an_b}(\omega_G) = \begin{pmatrix}
        \omega_G & 0 & \cdots & 0\\
        0 & \omega_G & \cdots & 0\\
        \vdots & \vdots & \ddots& \vdots\\
        0 & 0 & \cdots & \omega_G
    \end{pmatrix}.
\end{align}
Thus, the code encodes $2n_an_b$ logical qubits, and each row of $L_{X/Z}$ represents a logical operator supported on a distinct cluster. Comparing Eqs.~\eqref{Eq:PPSforCCcodes} and \eqref{Eq:LogicalDistribution}, we see that the nonzero entries of $H_a$ and $H_b$ determine which canonical logical $X$ operators are included in the measured logical subspace~\cite{gu2026qgpuparallellogicquantum}.

\vspace{5pt} \noindent {\bf{PPS for \emph{lifted} CC codes}.}
We now consider PPS for lifted CC codes. In general, lifted CC codes do not retain the algebraic structure that defines the CC family. Consequently, the relations that make PPS addressable for CC codes, such as Eqs.~\eqref{Eq:KernelForPPS} and~\eqref{Eq:LogicalDistribution}, do not automatically extend to the lifted code. Nevertheless, we show that, under suitable conditions, these relations can be transferred to the lifted code through the chain maps. We further show that all lifted CC codes listed in Table~\ref{Table:liftedCCcodes} satisfy these conditions and therefore retain the same logical addressability under PPS as their base CC codes.

Let $\widetilde{Q}$ be an odd index $t$ lift of a CC code $Q$, and assume that the number of logical qubits remains the same after the lift. Under these assumptions, the induced transfer map on the first homology is an isomorphism as discussed in Sec.~\ref{Subsec:ChainCochainMaps}. Therefore, a complete set of logical representatives for the lifted code is obtained by applying the transfer map to the canonical logical representatives of the base code.

Thus, the clustered support structure in Eq.~\eqref{Eq:LogicalDistribution} is preserved under this transfer. Specifically, each cluster of size $p$ in the base code is lifted to a distinct cluster of size $tp$ (c.f. Eq.~\eqref{Eq:transfermap_def}). Hence, the logical operators of the lifted CC code admit the canonical form
\begin{align}\label{Eq:LogicalDistribution_lift}
\widetilde{L}_{X/Z}
=
\operatorname{diag}_{2n_an_b}(\omega_H),
\end{align}
where $\omega_H = \sum_{h\in H}h$, and $H$ is the group of the lifted CC code.

A nontrivial issue is whether the lifted CC code $\widetilde{Q}$ acquires an extra nontrivial basis~\footnote{In fact, this corresponds to the $X$ metachecks in the single-shot QEC literature.} in $\ker(\widetilde{\partial}_2)$. The transfer map ${\tau}_2$ gives the inclusion:
\begin{align}\label{Eq:InclusionH2}
    {\tau}_2( \ker \partial_2 )\subseteq \ker \widetilde{\partial}_2.
\end{align}
In general, the inclusion in Eq.~\eqref{Eq:InclusionH2} may be strict: a lift can introduce additional $X$-metachecks that are not inherited from the base code. For example, a trivial disconnected lift produces additional independent copies of the metachecks of the base code. We show, however, that no such additional metachecks arise when the lift index is odd and the number of logical qubits is unchanged. This result is summarized in the following proposition.

\begin{proposition} Consider lifting a CC code $Q$ using a group $K$ of odd order. Suppose that the number of logical qubits is unchanged by the lift
    \begin{align} \label{eq:assumption1}
        \dim H_1(\widetilde{Q}) = \dim H_1(Q).
    \end{align}
    Then 
    \begin{align} \label{eq:desired_result1}
        \ker \widetilde{\partial}_2 = \tau_2(\ker \partial_2).
    \end{align}
\end{proposition}
\begin{proof}
    Since the CC code is obtained from the product of square matrices, the Euler characteristic $\chi_E$ of its product complex should vanish:
    \begin{align}
        \chi_E (Q_\bullet)=\dim Q_2-\dim Q_1+\dim Q_0=0.
    \end{align}
    The lifted complex has the same structure, so we have $\chi_E(\tilde{Q}_\bullet) = 0$. By the Euler--Poincar\'e identity,
    \begin{align}
        \dim H_2(Q)-\dim H_1(Q)+\dim H_0(Q)=0,
    \end{align}
    and similarly for $\widetilde{Q}_\bullet$. Thus
    \begin{align}  \label{eq:Euler_intermediate}
        \dim H_1(Q)&=\dim H_0(Q)+\dim H_2(Q), \nonumber \\
        \dim H_1(\widetilde{Q})&=\dim H_0(\widetilde{Q})+\dim H_2(\widetilde{Q}).
    \end{align}
    By assumption Eq.~\eqref{eq:assumption1}, the number of logical qubits is unchanged after the lift. Using Eq.~\eqref{eq:Euler_intermediate},
    \begin{align}\label{Eq:Homology_lifted_base_CC}
        \dim H_0(\widetilde{Q})+\dim H_2(\widetilde{Q})
        =
        \dim H_0(Q)+\dim H_2(Q).
    \end{align}
    Because the lift index is odd, the induced transfer map $\hat{\tau}_i$ is injective for $i=0,1,2$. In particular,
    \begin{align}\label{Eq:HomologyIneq}
        \dim H_i(\widetilde{Q})\geq \dim H_i(Q) \quad \textrm{for }i=0,1,2.
    \end{align}
    Comparing Eq.~\eqref{Eq:HomologyIneq} with Eq.~\eqref{Eq:Homology_lifted_base_CC}, we see that the inequalities for $i = 0, 2$ must be equalities. In particular,
    \begin{align}
        \dim H_2(\widetilde{Q})=\dim H_2(Q).
    \end{align}
    For odd lift index $t$, the transfer map is injective on homology, and together with Eq.~\eqref{Eq:InclusionH2}, we obtain the desired result Eq.~\eqref{eq:desired_result1}.
\end{proof}
For the lifted CC instances reported in Sec.~\ref{Sec:CodeSearch}, the number of logical qubits stays the same and thus we have
\begin{align}\label{Eq:LiftedXmetachecks}
    \ker(\widetilde{\partial}_2) = \tau_2 \ker(\partial_2) = \mathsf{row} (\omega_H I_a\otimes I_b).
\end{align}

As a corollary, we obtain the following.

\begin{corollary}[PPS for lifted CC codes]
Under the assumptions of the preceding proposition, the PPS of the base CC code induces a PPS of the lifted CC code.
\end{corollary}
\begin{proof}
    This follows from Eqs.~\eqref{Eq:LogicalDistribution_lift} and \eqref{Eq:LiftedXmetachecks}.
\end{proof}

\noindent {\bf Merged code distance.} There is no theoretical guarantee that the merged code distance is preserved by PPS for general LP codes~\cite{gu2026qgpuparallellogicquantum}. We therefore verify this property numerically for the lifted codes listed in Table~\ref{Table:liftedCCcodes}. For $k = 8$, no dressed data logical operator of weight below the data code distance was found in the randomized search for every binary choice of $H_a$ and $H_b$.

\section{Lift branches and finite-size crossings}\label{Sec:CI}

In this section, we study the critical behavior of the lifted family of LP codes under decoherence. It has been pointed out that topological memories exhibit a decoherence-induced phase transition (DIPT), which is captured through information-theoretic quantities such as coherent information~\cite{PRXQuantum.5.020343, hlfh-86yz, PhysRevResearch.6.L042014, fx56-8nvy, PhysRevA.111.032402, vijay2025informationcriticalphasesdecoherence}. Here, we study DIPT in algebraically defined qLDPC codes.

A central issue is that phase-transition behavior is sharply defined only in the thermodynamic limit, which is accessed through a sequence of increasing system sizes, which we refer to as a \emph{thermodynamic family}. For codes embedded in Euclidean space with a natural translation symmetry, such as toric codes, a thermodynamic family is obtained simply by enlarging the lattice. For algebraically defined qLDPC codes, however, the appropriate notion of a thermodynamic family is less clear because these codes generally lack an underlying geometrically local lattice.

\begin{figure}
    \centering
    \includegraphics[width=1.0\linewidth]{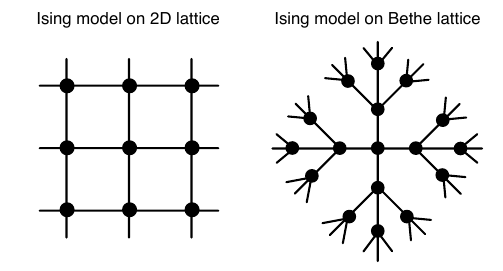}
    \caption{\textbf{Square lattice and a degree-$4$ Bethe lattice.} In both models, the interactions are pairwise and each spin participates in four interactions. The two models have the same radius-one local structure but different larger finite-radius neighborhood structures. The critical temperature is $1/\beta_c^{\rm 2D} = 2.27$ while $1/\beta_c^{\rm Bethe} = 2.89$ with unit coupling strength. }
    \label{Fig:Isingon2DandBethe}
\end{figure}

How should one construct a thermodynamic family from a single finite LP code? A natural possibility is to use graph lifts, which increase the code size while preserving the local incidence structure of the Tanner graph. Local incidence data alone, however, do not determine the thermodynamic limit. A familiar example is provided by the ferromagnetic Ising model on the two-dimensional square lattice and on the degree-$4$ Bethe lattice, illustrated in Fig.~\ref{Fig:Isingon2DandBethe}. The two models have the same radius-one Tanner-graph structure: each interaction couples two spins, and each spin participates in four interactions. Their larger neighborhoods and global geometries are nevertheless different, leading to distinct transition temperatures and critical behavior~\cite{PhysRev.65.117,Baxter1985}. Defining a thermodynamic family therefore requires more than preserving the immediate local connectivity.

A lift contains algebraic and combinatorial information beyond the radius-one neighborhood of the base Tanner graph. Its larger-radius structure is determined jointly by the group extension, the quotient map $\pi$, and the assignment of lifted group elements to the nonzero boundary-map terms. It is therefore natural to ask whether particular sequences of lifts define meaningful thermodynamic families and, conversely, which choices of lifts lead to distinct asymptotic behavior.

To begin addressing these questions, we perform an exploratory study of decoherence-induced phase transitions in several collections of finite lifts. We examine whether their normalized coherent-information curves develop stable finite-size crossings and how the observed behavior depends on the base code and the choice of group extension. For a code encoding $k$ logical qubits, let the encoded state be maximally entangled with a $k$-qubit reference system. The coherent information after the action of a noise channel quantifies how much of this encoded entanglement remains accessible~\cite{PhysRevA.54.2629,Horodecki2006}. It cannot increase under subsequent quantum channels by the data-processing inequality. Moreover, the coherent information retains its noiseless value $k$ if and only if there exists a recovery channel that restores the encoded state and its entanglement with the reference~\cite{nielsen2010quantum,PhysRevA.54.2629}. Preservation of coherent information therefore provides a necessary and sufficient condition for exact quantum error correction on the code space.

\begin{figure}[t!]
    \centering
    \includegraphics[width=1.0\linewidth]{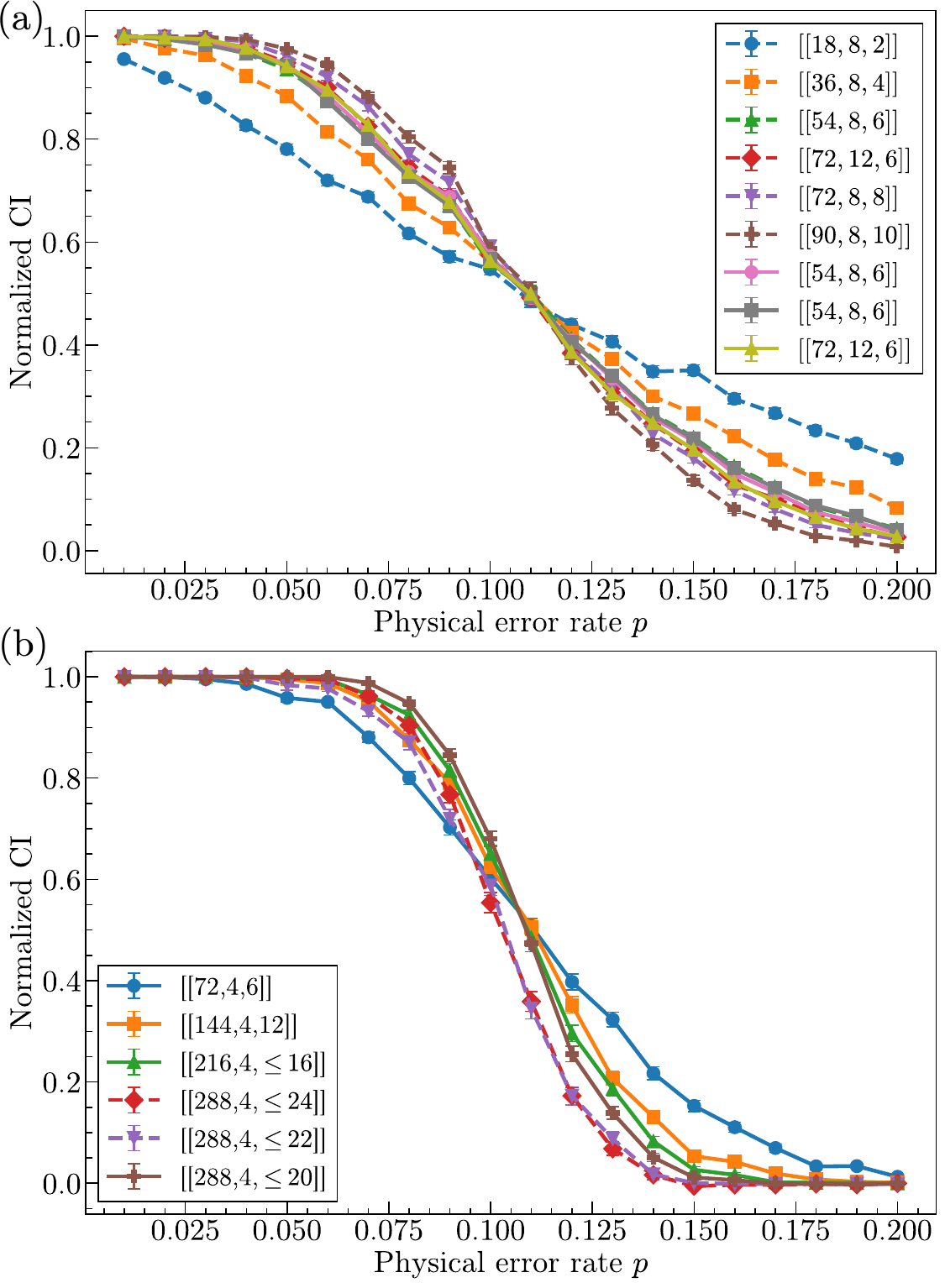}
    \caption{
    {\bf Coherent information of a base code and its lifts.}  The coherent information is normalized by the number of logical qubits $k$ so that we can compare codes with different $k$. 
    % (CI) for the lifted family of an LP code. 
    % Insets magnifies the crossing point. 
    {\bf (a)} We lift a BB base code with the code parameter $\llbracket18, 8,  2\rrbracket$ to obtain the codes in Table~\ref{Table:LiftedLPforCI} (the upper block). Dotted curves correspond to the codes that remain within the BB family, whereas solid curves correspond to non-BB lifted codes.
    {\bf (b)} We lift the non-abelian base code in Example~\ref{Ex:NonAbelianLP} to obtain the family in Table~\ref{Table:LiftedLPforCI} (the lower block). Solid curves correspond to the codes that show a crossing, whereas dotted curves correspond to those that do not.
    }
    \label{fig:CI_FiniteSizeScaling}
\end{figure}

\vspace{5pt} \noindent {\bf Finite-size crossing.} The observable we choose to study \emph{mixed-state phases}~\cite{Lee2025symmetryprotected, LeePRXQ2023, PRXQuantum.5.020343, vijay2025informationcriticalphasesdecoherence, yangPRX2026} of the LP code is coherent information~\cite{PhysRevA.54.2629} under Pauli errors. In this setting, the coherent information can be expressed in terms of a disordered classical statistical-mechanical model~\cite{hlfh-86yz, PhysRevA.111.032402}. We use this mapping and a Markov chain Monte Carlo (MCMC) method~\cite{Metropolis1953, Hastings1970} described in Appendix~\ref{App:SecCI}. Specifically, we combine standard Metropolis updates with parallel tempering~\cite{doi:10.1143/JPSJ.65.1604} to accelerate thermalization. For each disorder realization, we simulate $60$ replicas, discard the first $10{,}000$ Metropolis updates for thermalization, and then perform $100{,}000$ additional updates for measurement. The temperature ladder is optimized using the adaptive procedure of Ref.~\cite{Vousden2015}, and the final results are averaged over $1{,}000$ independent disorder realizations.

First, we study a family of LP codes obtained by lifting a BB code. We choose a base code defined for the ring $\F_2[\Z_3\times \Z_3]$, with the boundary maps
\begin{align}
    a = 1 + y + y^2, \quad b = 1 + x + x^2.
\end{align}
The code parameter is $\llbracket18, 8, 2\rrbracket$~\cite{symons2025sequences}. We lift the base code with the groups $K = \Z_2, \Z_3, \Z_4, \Z_2\times \Z_2, \Z_5$. The resulting lifted codes are listed in the upper block of Table~\ref{Table:LiftedLPforCI}.
BB codes can be represented as two-dimensional translationally invariant stabilizer codes~\cite{Haah2013, rmy6-9n89}. Therefore, a conventional thermodynamic family can be obtained by replacing $\mathbb Z_3\times\mathbb Z_3$ with increasingly large translation groups $\mathbb Z_n\times\mathbb Z_m$. However, the lifts considered here are more general: several of the codes in Table~\ref{Table:LiftedLPforCI} are not members of such a translationally invariant BB family, even though they are obtained by lifting the same base code.

The results are shown in Fig.~\ref{fig:CI_FiniteSizeScaling}(a). Because the maximum value of coherent information is $k$, which is the number of logical qubits, we plot the normalized coherent information $I_C/k$ to compare codes with different values of $k$. 
As the code size increases, the curves cross near $p_c\simeq0.11$, which is the critical value observed for $\Z_2$ toric code under independent $X$ noise~\cite{hlfh-86yz}. 
For the BB members of the family, this agrees with the interpretation of BB codes as multiple coupled copies of the toric code~\cite{86j7-cmsw, Haah2013, 10.1063/5.0021068, Bombin2012, Haah2017, Bombin2014}. 
Across the finite lifts studied here, the normalized coherent-information curves exhibit an apparent common crossing. This is compatible with the toric-code threshold, including for several non-BB lifts. However, because these codes do not yet constitute a structurally specified thermodynamic sequence, the present data do not establish a common asymptotic critical point.

\newcommand{\nonabelian}[1]{%
    \textcolor{red!70!black}{%
        \ensuremath{#1}%
    }%
}
\begin{table}
\centering
\setlength{\tabcolsep}{2pt}      % default is about 6pt
\renewcommand{\arraystretch}{1.05}
\begin{tabular}{c|c|c|c}
\toprule
$|K|$ & Group $K$ & Group $H$ & $\llbracket n,k,d\rrbracket $\; \\
\midrule
1 & $\{e\}$ & $\Z_3\times\Z_3$
& $\llbracket18,8,2\rrbracket$\\

2 & $\Z_2$ & $\Z_6\times\Z_3$
& $\llbracket36,8,4\rrbracket$\\

3 & $\Z_3$ & $\Z_9\times\Z_3$
& $\llbracket54,8,6\rrbracket$\\

4 & $\Z_2\times\Z_2$ & $\Z_6\times\Z_6$
& $\llbracket72,12,6\rrbracket$\\

4 & $\Z_4$ & $\Z_{12}\times\Z_3$
& $\llbracket72,8,8\rrbracket$\\

5 & $\Z_5$ & $\Z_{15}\times\Z_3$
& $\llbracket90,8,10\rrbracket$\\

3 & $\Z_3$
& \nonabelian{(\Z_3\times\Z_3)\rtimes\Z_3}
& $\llbracket54,8,6\rrbracket$\\

3 & $\Z_3$
& \nonabelian{\Z_9\rtimes\Z_3}
& $\llbracket54,8,6\rrbracket$\\

4 & $\Z_2\times\Z_2$
& \nonabelian{\Z_3\times A_4}
& $\llbracket72,12,6\rrbracket$\\
\midrule
1 & $\{e\}$ &$S_3 \times S_3$ & $\llbracket 72, 4, 6\rrbracket$\\
2 & $\Z_2$ &$(\Z_3\rtimes\Z_4)\times S_3$ & $\llbracket144, 4, 12\rrbracket$\\
3 & $\Z_3$ &$((\Z_3\times\Z_3)\rtimes\Z_3)\rtimes (\Z_2\times\Z_2)$ & $\llbracket216, 4, \leq 16\rrbracket$\\
4 & $\Z_4$ &$(\Z_3\times\Z_3)\rtimes {\mathrm{Q}}D_{16}$ & $\llbracket288, 4, \leq 22\rrbracket$\\
4 & $\Z_4$ &\nonabelian{(\Z_3\rtimes\Z_8)\times S_3} & $\llbracket288, 4, \leq 24\rrbracket$\\
4 & $\Z_4$ &\nonabelian{(\Z_3\times\Z_3)\rtimes (\Z_8\rtimes\Z_2)} & $\llbracket288, 4, \leq 22\rrbracket$\\
\bottomrule
\end{tabular}
\caption{\textbf{
The lifted LP codes studied in Sec.~\ref{Sec:CI}.} We calculate the coherent information for these codes and study the finite-size crossing. The codes in the upper block are obtained from a BB code as a base code. The boundary maps of the base code are $a = 1 + y + y^2, b = 1 + x + x^2$ over the polynomial ring $\F_2[x, y]/(x^3 - 1, y^3 - 1)$. Non-BB instances are highlighted with the red color. The codes in the lower block are obtained from the base code in Example~\ref{Ex:NonAbelianLP}.
We extend the group $G$ to the group $H$ using the group $K$ shown in the table. The codes that do not show crossings are highlighted with the red color.
}
\label{Table:LiftedLPforCI}
\end{table}

Second, we study the family obtained by lifting the non-abelian LP code in Example~\ref{Ex:NonAbelianLP}. We lift the base code with the kernels $K = \Z_2, \Z_3, \Z_4$, and the resulting codes are provided in Table~\ref{Table:LiftedLPforCI}. As in the Abelian case, we consider independent $X$ errors and evaluate the normalized CI. The results are shown in Fig.~\ref{fig:CI_FiniteSizeScaling}(b). 
Here, we have different behavior. 
The complete set of lifts does not exhibit a single common crossing. A subset of the curves intersects near $p \sim 0.105$, whereas the remaining extension choices shift the apparent crossing toward lower error rates. Since this subset has not been selected by an independent structural criterion, the observation does not define a thermodynamic branch. It instead demonstrates that the covering relation to a common base code is insufficient to determine common finite-size scaling.

\vspace{5pt} \noindent {\bf Remark.} Although some selected lifts exhibit well-defined crossings, this behavior is not determined by the local Tanner-graph incidence structure alone, just as identical local interaction rules can lead to different thermodynamic behavior on square and Bethe lattices. For the statistical-mechanical models associated with decohered CSS codes~\cite{hlfh-86yz,PhysRevA.111.032402}, different group extensions can likewise produce models with distinct patterns of ${\cal O}(1)$-neighborhood interaction graphs and different scaling of $(k,d)$/large-scale domain-wall structures. Some lifts may remain analogous to multiple copies of the toric code, whereas others define genuinely different code families~\cite{86j7-cmsw, Haah2013, 10.1063/5.0021068, Bombin2012, Haah2017, Bombin2014}. A physically meaningful thermodynamic family may therefore require additional constraints on finite-radius and global geometric properties that select a particular branch among the possible lifts of a fixed base code. We leave the identification of such conditions to future work.

\section{Conclusion and outlook}\label{Sec:conclusion_outlook}

In this work, we developed a systematic algebraic procedure for lifting a base LP code to larger LP codes. The construction extends the underlying group and compatibly lifts the two constituent classical codes. Here we explored three applications:
\begin{itemize}[leftmargin=13pt]
    \item \textbf{Code construction.}
    A search over the resulting lifts yields several LP codes whose parameters improve on previously reported instances.

    \item \textbf{Logical operations.}
    The chain maps induced from the lifting allow logical-operation gadgets to be transferred from a projected code to its lifts. In \emph{lifted code surgery}, an ancilla constructed for the smaller projected code is coupled to the lifted code through composite chain maps, which was further boosted to preserve the distance. The result reduced the ancillary space overhead in the examples studied. We also give conditions under which lifted CC codes inherit parallel product surgery from their base codes, thereby retaining their logical addressability.

    \item \textbf{Thermodynamic families.} The coherent-information curves exhibit finite-size crossings for selected lifts, consistent with common thermodynamic behavior. The crossings are not universal across all lift branches, however: a common base code and identical local incidence structure do not determine a unique thermodynamic family. Different lifts can have distinct global geometries, code parameters, and critical behavior.
\end{itemize}

Our work leaves several interesting directions. 

First, how do logical degrees of freedom transform under a lift? For odd-index lifts, the transfer maps identify logical operators inherited from the base code, whereas no analogous relation follows directly for even-index lifts. A structural description that separates inherited logical operators from those created by the lift would sharpen parameter bounds and provide explicit logical bases. Relating the algebraic construction developed here to the Freedman--Hastings geometric framework may offer a useful route to such a description~\cite{guemard2025liftingcsscodehandlebody,freedman2021buildingmanifoldsquantumcodes}.

Second, which logical gadgets can utilize this lifting structure? Our results for code surgery and PPS suggest that the induced chain and cochain maps may provide a route to efficiently implement other logical-operation gadgets. One main question here is how to preserve the distance in this inherited gadget. The other question is to determine when transversal gates, homomorphic CNOTs and measurements, or code-switching protocols for a base code induce corresponding operations on its lifts~\cite{Breuckmann2024,vf7v-cpq9,berthusen2025automorphismgadgetshomologicalproduct,Breuckmann2026,ghhp-cytl,li2026transversaldimensionjumpproduct,davenport2026generalizedbicyclecodescyclic,10812769,Xu2025,PRXQuantum.4.030301}. Addressing this question requires understanding the action of the transfer maps on explicit logical bases, for which algebraic descriptions of LP-code logical operators may be useful~\cite{davenport2026generalizedbicyclecodescyclic,lee2026logicalspectroscopyliftedproductcodes}.

Third, which branches of lifts define the same thermodynamic family? Our coherent-information results indicate that a common base code and identical local incidence data are insufficient: different lifts can have different larger finite-radius neighborhoods, global connectivity, and domain-wall geometry, leading to distinct critical behavior. A classification of lift branches must therefore incorporate geometric information beyond the base Tanner graph. Natural candidates include finite-radius neighborhood growth, the intrinsic locality dimension~\cite{lu2026intrinsiclocalitydimensionquantum}, and global expansion properties such as Cheeger constants. These quantities may control not only code parameters but also the domain-wall free-energy landscape and the emergence of topological spin-glass order~\cite{placke2025expansioncreatesspinglassorder,placke2024topologicalquantumspinglass}. Since different lifts of the same graph can have sharply different expansion properties~\cite{AMIT_LINIAL_2006,bordenave2019newprooffriedmanssecond}, lift families provide a controlled setting in which to determine which combination of local, geometric, and homological data governs both asymptotic code performance and decoherence-induced phases.

Several more practical questions can be pursued within the same setting, including whether decoders and hardware-efficient implementations transfer from a base code to its lifts~\cite{gu2026scalableneuraldecoderspractical,tan2026generalizedmatchingdecoders2d,Xu2024}, and whether the construction extends to other high-rate group-algebra codes and quantum Tanner codes~\cite{okada2026highgirthregularquantumldpc,cain2026shorsalgorithmpossible10000,lee2026logicalspectroscopyliftedproductcodes,zhao2026ultrahighratequantumerrorcorrection,kasai2026breakingorthogonalitybarrierquantum,leverrier2022quantumtannercodes,radebold2025explicit}.

\begin{acknowledgments}

We thank Yu-An Chen, Ryohei Kobayashi, Nobuyuki Yoshioka, Akash Vijay, and Tai-Hsuan Yang for fruitful discussions.
YH acknowledges the support by Murata Overseas Scholarship and the KDDI Foundation Overseas Scholarship.
JYL acknowledges the support by the faculty startup grant at the University of Illinois, Urbana-Champaign and the IBM-Illinois Discovery Accelerator Institute. 

\emph{Note added}: Upon completion of this work, we noticed a recent work~\cite{aydin2026breakingbicycleframecosetbased} that also found the lifted BB code instance $\llbracket128, 16, 12\rrbracket$ in Table~\ref{table:Lifted_BBcodes}.

\end{acknowledgments}

\appendix

\section*{Data availability}

The boundary maps for all the lifted codes obtained in this work are available from the corresponding author upon reasonable request. We also provide the boundary maps of some of the codes in the Appendix~\ref{App:boundarymaps} when the group has a simple structure.

\section{Proofs of the propositions}
We provide detailed proofs of the theorems in the main text.
\subsection{Proof of Proposition~\ref{theorem:liftingLPbyGroupExtensions}}\label{Sec:ProvingLiftingTheorem}

We complete the proof of Proposition~\ref{theorem:liftingLPbyGroupExtensions}.

\begin{proof} {\bf (1).} 
    We first construct the graph projection $p$ induced by the group homomorphism $\pi$. Let $\mathcal{T}(Q)$ denote the Tanner graph of code $Q$. A vertex of $\mathcal{T}(\widetilde{Q})$ is labeled by a pair $(h,\mu)$, where $h\in H$ and $\mu$ labels a basis element of the corresponding chain module; for example for an $X$-check vertex, a basis element $\mu \in \widetilde{A}_1\otimes \widetilde{B}_1$. Similarly, a vertex of $\mathcal{T}(Q)$ is labeled by $(g,\mu)$, where $g\in G$. See Fig.~\ref{Fig:LiftbySES} for illustration.  We define the vertex projection by
    \begin{align}
        p_V:V(\mathcal{T}(\widetilde{Q}))
        &\longrightarrow
        V\left(\mathcal{T}({Q})\right), \nonumber \\
        (h,\mu)
        &\longmapsto
        (\pi(h),\mu),
    \end{align}
    where $V(\mathcal{T}(Q))$ denotes the set of vertices in $\mathcal{T}(Q)$.

   We show that $p_V$ preserves adjacency. Namely, if two vertices $v_1$ and $v_2$ are connected by an edge in $\widetilde{Q}$, then their images $p_V(v_1)$ and $p_V(v_2)$ are connected by an edge in $Q$. Each edge of $Q$ is specified by each group element appearing with nonzero coefficient in an entry of the boundary maps. 
   Fix an entry $a_{\mu\nu}\in\mathbb F_2[G]$ of one of the boundary maps and write
    \begin{align}
    a_{\mu\nu}&=
    \sum_{s\in G}c_{\mu\nu}(s)s, \quad c_{\mu \nu}(s) \in \mathbb{F}_2, \nonumber \\
    S_{\mu\nu}
    &:=
    \operatorname{supp}(a_{\mu\nu})
    =
    \left\{
    s\in G \, | \,
    c_{\mu\nu}(s)\neq 0
    \right\}.
    \end{align}
    Here $a_{\mu\nu}$ is simply the sum of the elements in $S_{\mu\nu}$. Each $s\in S_{\mu\nu}$ gives a separate family of edges in the binary expansion of the boundary map, which we call a \emph{supported group element}.

    Consider first an entry acting by left multiplication. For every $g\in G$ and $s\in S_{\mu\nu}$, the corresponding base edge is $e(g,\mu,\nu;s): (g,\mu) \longrightarrow (sg,\nu)$.
    For each $s\in S_{\mu\nu}$, choose a kernel label $\gamma_{\mu\nu}(s)\in K$ and define
    \begin{align}
    t_{\mu\nu}(s)
    &:=
    \chi\left(\gamma_{\mu\nu}(s)\right)\sigma(s)
    \in H.
    \end{align}
    The corresponding entry of the lifted boundary map is
    \begin{align} \label{eq:lifted_boundarymap}
    \widetilde{a}_{\mu\nu}
    &=
    \sum_{s\in S_{\mu\nu}}t_{\mu\nu}(s) \in R_H.
    \end{align}
    By construction for every $s\in S_{\mu\nu}$, $\pi\left(t_{\mu\nu}(s)\right) = s$.
    The lifted edge is then defined as $\widetilde{e}(h,\mu,\nu;s): (h,\mu) \longrightarrow \left(t_{\mu\nu}(s)h,\nu\right)$ for every $h\in H$. Under the vertex projection,
    \begin{align}
    p_V\left(t_{\mu\nu}(s)h,\nu\right)
    =
    \left(
    \pi\left(t_{\mu\nu}(s)h\right),
    \nu
    \right) =
    \left(
    s\pi(h),
    \nu
    \right).
    \end{align}
    Thus $\widetilde{e}(h,\mu,\nu;s)$ projects to the base edge $e\left(\pi(h),\mu,\nu;s\right)$.
    The same argument applies to entries acting by right multiplication (important for non-abelian case), for which the base and lifted edges take the forms
    \begin{align}
    (g,\mu)
    &\longrightarrow
    (gs,\nu), \quad 
    (h,\mu)
    \longrightarrow
    \left(ht_{\mu\nu}(s),\nu\right),
    \end{align}
    respectively.  Therefore, every lifted edge projects to its corresponding base edge. This defines a graph map
    $p:\mathcal{T}(\widetilde{Q})\longrightarrow\mathcal{T}(Q)$.

    \vspace{5pt} \noindent {\bf (2).} Next we show that $p$ is surjective. Take a vertex $(g,\mu)$ in $\mathcal{T}(Q)$. Its preimage under $p_V$ is
    \begin{align}
    p_V^{-1}(g,\mu)
    =
    \left\{
    \left(F(\gamma,g),\mu\right) \,|\, \gamma\in K \right\},
    \end{align}
    where $F: K \times G \rightarrow H$ is the bijection defined in Eq.~\eqref{Eq:BijectionofSES}. Thus $p$ is surjective, and the cardinality of the preimage is exactly $|K|$ for all the vertices in $\mathcal{T}(Q)$.

    \vspace{5pt} \noindent {\bf (3).} For local bijectivity, let $(h,\mu)$ be a vertex of $\mathcal{T}(\widetilde{Q})$ and let $(g,\mu) = p_V(h,\mu)$.
    An edge incident to $(g,\mu)$ is specified by a position $(\mu,\nu)$ in a boundary map together with a supported group element $s\in S_{\mu\nu}$. 
    By construction in Eq.~\eqref{eq:lifted_boundarymap}, for each $s$ there is a unique $t_{\mu\nu}(s)$; it simply means that every edge incident to $(g,\mu)$ has a unique incident lift at $(h,\mu)$. Conversely, every edge incident to $(h,\mu)$ arises from one of these lifted supported terms and therefore projects to an edge incident to $(g,\mu)$.
    
    Thus $p$ induces a bijection between the edges incident to $(h,\mu)$ and those incident to $(g,\mu)$. Therefore, $p$ is locally bijective. Since every vertex of $\mathcal{T}(Q)$ has exactly $|K|$ preimages, the Tanner graph of $\widetilde{Q}$ is a $|K|$-cover of the Tanner graph of $Q$.
        \end{proof}

\subsection{Proof of Proposition~\ref{theorem:projectingLPbyGroupExtensions}} \label{Sec:ProvingProjectionTheorem}

    \begin{proof}
    We first verify the chain-map identities. Since all maps are $\F_2$-linear, it suffices to consider a basis vector. We first consider basis vectors in $Q_2$:
    \begin{align}
    e_i^A h\otimes e_j^B
    \in
    Q_2,
    \qquad
    h\in H,
    \end{align}
    where $e_i^{A}$ $(e_j^B)$ is a basis of $A_1 \; (B_1)$. We also denote the basis for $A_0 \; (B_0)$ as $f_i^A$ $(f_j^B)$. By the definition of $\partial_2$,
    \begin{align}\nonumber
        \partial_2&\left(e_i^A h\otimes e_j^B\right) = \\
        &\sum_k f_k^A(\partial_A)_{ki}h\otimes e_j^B + \sum_k e_i^A\otimes h(\partial_B)_{kj}f_k^B.
    \end{align}
    Applying $p_1$ gives
    \begin{align}
        p_1\circ\partial_2\left(e_i^A h\otimes e_j^B\right) &= \sum_k f_k^A\Pi\left((\partial_A)_{ki}h\right)\otimes e_j^B \nonumber \\
        &\quad+ \sum_k e_i^A\otimes \Pi\left(h(\partial_B)_{kj}\right)f_k^B.
    \end{align}
    Since $\Pi$ is a ring homomorphism,
    \begin{align}
        \Pi\left((\partial_A)_{ki}h\right) &=\Pi\left((\partial_A)_{ki}\right)\Pi(h) = \left({\partial}_A^\pi\right)_{ki}\Pi(h),\\
        \Pi\left(h(\partial_B)_{kj}\right) &= \Pi(h)\Pi\left((\partial_B)_{kj}\right) = \Pi(h)\left({\partial}_B^\pi\right)_{kj}.
    \end{align}
    These are precisely the terms obtained by first applying $p_2$ and then the projected boundary map $\partial_2^\pi$. Hence,
    \begin{align}
        p_1\circ\partial_2 = \partial_2^\pi\circ p_2.
    \end{align}
    The other identity $p_0\circ\partial_1 =\partial_1^\pi\circ p_1$
    follows by the same argument. Therefore, $p_\bullet$ defines a chain map over $\F_2$.
    
    Next, we consider the cochain maps. 
    % This requires some care because the cochain complex is defined as the $\F_2$-dual of the chain complex, rather than as its naive $R$-linear dual. 
    The $\F_2$-dual LP complex can be expressed as the balanced product of the classical cochain complexes ${}$
    \begin{align}\label{Eq:Cochains}
        A^0&\xrightarrow{\partial_A^T} A^1, \qquad B^0 \xrightarrow{\partial_B^T} B^1,
    \end{align}
    whose coboundary maps are given by
    \begin{align}\label{Eq:CoboundaryMaps}
        \left(\partial_A^T\right)_{ji} &= \left(\partial_A\right)_{ij}^{\ast},&
        \left(\partial_B^T\right)_{ji} &= \left(\partial_B\right)_{ij}^{\ast}.
    \end{align}
    Here, $\ast:R_H\rightarrow R_H$ denotes the group-algebra involution
    \begin{align}
        \bigg(\sum_{h\in H}a_h h\bigg)^{\ast} = \sum_{h\in H} a_h h^{-1}.
    \end{align}    
% The canonical pullback of $\partial_A^T$ on the $R$-duals are
% \begin{align}
%     A^1\xlongrightarrow{{\partial_A^T}'}A^0,
% \end{align}
% where the boundary map is defined by 
% \begin{align}\label{Eq:RdualBdMap}
%     \left({\partial_A^T}'\right)_{ji} = \left(\partial_1^A\right)_{ij}.
% \end{align}
% Note that the involution does not appear here.
    
    % Let $f:R_H\rightarrow\F_2$ extract the coefficient of the identity element $f\left(\sum_h a_h h\right) = a_e$. The $R$-dual $a^1\in A^1$ is identified with the $\F_2$-dual via
    % \begin{align}\nonumber
    % \Theta: \operatorname{Hom}_{R_H} (A_1, R_H)\longrightarrow \operatorname{Hom}_{\F_2}(A_1, \F_2)\\
    % \Theta(a^1)(a_1)\coloneqq f(a^1(a_1))\in \F_2
    % \end{align}
    % for $a_1\in A_1$. Under this identification, the indicator functional dual to the $\F_2$-basis vector $e_i^A h$ is represented by $h^{-1}e_A^i$. Consequently, the $\F_2$-dual of a boundary matrix is obtained by transposing it and applying $\ast$ entrywise, yielding Eqs.~\eqref{Eq:CoboundaryMaps}. The same argument applies to $B_\bullet$, with the module handedness reversed.

    Taking the balanced product over $R$ of the cochain complexes in Eq.~\eqref{Eq:Cochains} gives
    \begin{equation}
        \begin{tikzcd}
            & B^0\otimes_R A^1 & \\
            B^0\otimes_R A^0 && B^1\otimes_R A^1\\
            & B^1\otimes_R A^0 &
            \arrow["{I\otimes\partial_A^T}", from=2-1, to=1-2]
            \arrow["{\partial_B^T\otimes I}", from=1-2, to=2-3]
            \arrow["{\partial_B^T\otimes I}", from=2-1, to=3-2]
            \arrow["{I\otimes\partial_A^T}", from=3-2, to=2-3]
        \end{tikzcd}.
    \end{equation}
    Under the identification above, this complex is naturally isomorphic to the $\F_2$-dual of the original LP chain complex.

    Specifically, let ${\partial}_A^T$ and ${\partial}_A^{\pi, T}$ denote the classical coboundary maps of $A^\bullet$ and $A_\pi^\bullet$, respectively. They satisfy ${\partial}_A^{\pi, T} = \Pi\left({\partial}_A^T\right)$,
    where $\Pi$ is applied entrywise. This follows from the fact that $\Pi$ commutes with the group-algebra involution. The analogous relation holds for ${\partial}_B^T$ and $\partial_B^{\pi, T}$. Therefore, by the same argument as in the chain-complex case, the projection maps commute with the coboundary maps and define a cochain map.

\end{proof}

\section{Details of code surgery}\label{App:surgery}
In this section, we provide further details of the surgery protocol. We first discuss the logical operators of the merged codes, whose minimal weight determines the phenomenological distance of code surgery.

Define the $i$-th homology group of a chain complex $Q_\bullet$ as $H_i(Q):= \ker \partial_i / \Im \partial_{i+1}$. Then the logical $X$ operators of the merged code in Eq.~\eqref{eq:surgery_diagram_base} are given by~\cite{weibel1994introduction} 
\begin{align}\label{Eq:MergedCodeLogicalX}
    H_1(\operatorname{cone} (\Gamma)) \simeq \left[\frac{H_1(Q)}{ \hat{\Gamma}_1[ H_1(A)] }\right]\oplus \ker(H_0(A)\xrightarrow[]{\hat{\Gamma}_0}H_0(Q)).
\end{align}
The first term represents the logical $X$ operators of the data code that remain after the measurement, while the second term represents additional logical operators originating from the ancilla code, which may be regarded as gauge logical operators of the surgery construction.

The physical interpretation of the second term is as follows. Among the logical $X$ operators of the ancilla code $A$, those mapped to the trivial class in $H_0(Q)$ under $\hat{\Gamma}_0$ are promoted to logical $X$ operators in the merged code. Let $l_X^A$ be a logical $X$ operator of $A$. Its homology class lies in the kernel of $\hat{\Gamma}_0$ precisely when
\begin{align}\label{Eq:App_gaugelogical}
    \Gamma_0 l_X^A \in \Im \partial_1\Longleftrightarrow \exists v\in Q_1 \,\,\,{\rm s.t.} \,\, \Gamma_0 l_X^A = \partial_1 v.
\end{align}
The merged $Z$-check matrix is 
\begin{align}
    \begin{pmatrix}
        \partial_0^A & 0\\
        \Gamma_0 & \partial_1
    \end{pmatrix}.
\end{align}
Equation~\eqref{Eq:App_gaugelogical}, together with $\partial_0^A l_X^A=0$, implies that $(l_X^A,v)$ commutes with all merged $Z$ checks. Thus, Eq.~\eqref{Eq:App_gaugelogical} is exactly the condition under which an ancilla logical operator can form a logical operator of the merged code by dressing a product of Pauli $X$ operators on the data code.

Similarly, consider the dual complex obtained by transposing the boundary maps and the chain map of the original complex:
\begin{equation}
\label{eq:surgery_diagram_base_cochain}
    \begin{tikzcd}
    {Q^0} & {Q^1} & {Q^2} \\
    {A^{-1}} & {A^0} & {A^1} 
    \arrow["{\partial^0}", from=1-1, to=1-2]
    \arrow["{\partial^1}", from=1-2, to=1-3]
        \arrow["{{\partial}^{-1}_A}", from=2-1, to=2-2]
    \arrow["{{\partial}^0_A}", from=2-2, to=2-3]
        \arrow["{\Gamma^0}", from=1-1, to=2-2]
    \arrow["\Gamma^1", from=1-2, to=2-3]
    \end{tikzcd}.
\end{equation}
The logical $Z$ operators of the merged code are characterized by
\begin{align}
    H^1(\operatorname{cone} (\Gamma)) \simeq \left[\frac{H^0(A)}{\hat{\Gamma}^0 [H^0(Q)]}\right]\oplus \ker(H^1(Q)\xrightarrow[]{\hat{\Gamma}^1}H^1(A)).
\end{align}
Here, the first term represents the gauge logical operators originating from the ancilla code, while the second term represents the logical $Z$ operators of the data code that remain after the measurement.

The two terms admit an interpretation analogous to that for logical $X$ operators, with the roles of the data and ancilla codes exchanged. Specifically, the unmeasured gauge $Z$ logical operators of the ancilla code survive as logical operators of the merged code~\cite{cross2025improvedqldpcsurgerylogical}, while the surviving data-code logical $Z$ operators are dressed by Pauli $Z$ operators supported on the ancilla code. In particular, we see that the distance of the data logical $Z$ operators should always be at least $d_Z$~\cite{PhysRevX.15.021088, zheng2025highratesurgeryconstantoverheadlogical}.

The non-trivial issue is the distance of the data logical $X$ operators, the first term in Eq.~\eqref{Eq:MergedCodeLogicalX}. A representative has a weight $d_X$, but one can deform this by multiplying it by merged $X$-checks and gauge logical $X$ operators, potentially producing an equivalent representative with weight smaller than $d_X$. Several works have been conducted to address this issue~\cite{doi:10.1126/sciadv.abn1717, cross2025improvedqldpcsurgerylogical, Williamson2026, PhysRevX.15.021088, zheng2025highratesurgeryconstantoverheadlogical, 1g44-jp62}. When the merged code has no gauge logical operators, it is sufficient to improve the expansion properties of the boundary map $\partial_1^A$~\cite{Williamson2026, 1g44-jp62, PhysRevX.15.021088}. When gauge logical operators are present, the analysis is more involved, but one can still guarantee preservation of the merged-code distance in some cases~\cite{doi:10.1126/sciadv.abn1717, zheng2025highratesurgeryconstantoverheadlogical}. Although these approaches provide sufficient expansion conditions, in practice it is often enough to add qubits and checks randomly until the desired merged-code distance is preserved~\cite{zheng2025highratesurgeryconstantoverheadlogical}. We adopt the latter approach in constructing the ancilla codes.

\subsection{Low rate surgery}\label{App:LowRateSurgery}

In the ordinary low-rate (graph) surgery, the initial ancilla graph is constructed by the standard gauging and perfect-matching procedure~\cite{Williamson2026}.

Our heuristic approach, namely \emph{lifted graph surgery}, uses this established protocol: we initialize an ancillary code for the base code, and takes the composite map with the transfer map. After initialization, we randomly add edges until the merged code distance is preserved. While this approach can only measure logical operators transferred from the logicals of the base code, this approach can reduce the ancilla overhead by a constant factor.

The key to reducing the ancilla overhead is to choose a low-weight representative of the logical operator to be measured. We search for such a representative using a randomized algorithm based on random permutations and Gaussian elimination. 

Suppose that the target is a logical $X$ operator. We first construct a matrix $G$ whose rows consist of the $X$-type stabilizer generators together with a representative of the target logical $X$ operator. We then permute the columns of $G$ by a random permutation $\sigma$, perform Gaussian elimination, and apply the inverse permutation $\sigma^{-1}$ to the resulting columns to obtain a new generating matrix $G'$. Each row of $G'$ lies in the span of the target logical operator and the $X$-type stabilizers. We determine whether a row represents the target logical class by computing its pairing with the conjugate logical $Z$ operator, and among the valid representatives, we select the one with the smallest weight. We then construct the ancilla code using this optimized representative to obtain the results reported in Tables~\ref{tab:q41-wy-overhead} and~\ref{tab:q41-wy-x1-distributions}. Note that the low-weight representatives are chosen independently for the lifted code surgery and direct code surgery.

\subsection{High rate surgery}\label{App:HighRateSurgery}

The algorithm is based on the randomized construction described in Ref.~\cite{zheng2025highratesurgeryconstantoverheadlogical}; see Appendix~E of its Supplemental Material. In direct code surgery, the ancilla code and the associated chain maps are constructed directly for the lifted data code. In lifted code surgery,  the ancilla code is first constructed for the base code, and its chain maps are composed with the transfer maps to couple the ancilla to the lifted code. The two procedures differ only in their initialization. After initialization, both use the same expansion-enhancement procedure; additional $Z$ checks are introduced randomly to eliminate low-weight representatives of the surviving data logical $X$ operators. These data logical operators correspond to the first term in Eq.~\eqref{Eq:MergedCodeLogicalX}, whereas the logical operators represented by the second term are treated as gauge $X$ operators.

Empirically, we find that choosing suitable representatives of the measured logical operators is crucial to the success of the algorithm. A representative is considered good if (i) the union of the supports of the measured logical operators is as small as possible and (ii) this union supports as few unmeasured logical operators as possible, ideally none. We search for such representatives using a heuristic logical-basis search described below.

We construct a matrix $G$ whose rows consist of the $m$ logical $X$ operators to be measured together with the $r_X$ $X$-type stabilizer generators. After applying a random column permutation and Gaussian elimination, we obtain a matrix $G'$. We then compute the pairing of every row of $G'$ with the $m$ conjugate logical $Z$ operators. These pairings form a binary matrix $G''$ of size $(m+r_X)\times m$, where each row of $G''$ records the logical class represented by the corresponding row of $G'$. We order the rows of $G'$ by increasing weight and greedily select rows whose corresponding rows in $G''$ span the full $m$-dimensional space. This produces a low-weight set of representatives for the measured logical subspace.

To address the second criterion, we restrict the $Z$-stabilizer matrix to the union of the selected supports. Vectors in the kernel of this restricted matrix correspond to Pauli $X$ operators supported within the union that commute with all $Z$ checks, which should be either a product of $X$-checks or logical $X$ operators. We randomly sample such vectors and determine whether they represent nontrivial logical $X$ operators outside the measured logical subspace. Repeating this procedure for different randomized basis choices, we select the representatives whose combined support contains the fewest sampled unmeasured logical operators.

We select the logical representatives independently for lifted and direct surgery. For lifted surgery, we minimize the weight of the representatives in the base code while also minimizing the number of unmeasured logical operators supported on their transferred support in the lifted code. For direct surgery, both criteria are evaluated directly in the lifted code.

\section{Details of the numerics}
\subsection{Code search}\label{Appendix:DetailsofNumerics_codesearch}

We perform the numerical simulations for the code search using the programming language \texttt{GAP}. Most of the subroutines are discussed in the main text. We discuss a technique to reduce the search space below.
The complete procedure is summarized in Algorithm~\ref{alg:extension-lift-search}.

As explained in Sec.~\ref{Sec:LiftingLPcodes}, a lift is specified by assigning an element of $K$ to every supported term in the two boundary maps for $A$ and $B$. For example, consider the BB code
\begin{align}\label{Eq:sample_BB}
    a
    &=
    1+x+y^2,
    &
    b
    &=
    1+y+x^2.
\end{align}
Since the two polynomials contain six supported elements in total, a naive search examines $|K|^6$ assignments. Many of these assignments, however, produce isomorphic lifted codes. Our goal is to remove part of this \emph{gauge redundancy} before performing the search.

The lift of an LP code is constructed by first lifting the two classical codes $A_\bullet$ and $B_\bullet$ and then taking their balanced product. Accordingly, we perform the gauge fixing on $A_\bullet$ and $B_\bullet$ separately. We then show that the gauge equivalent classical complexes yield isomorphic balanced-product LP codes.

We first consider a lift of the classical code $A_\bullet$. Here we introduce the notion of protograph, which is distinguished from the Tanner graph. The protograph $\mathcal P_A$ is the labeled bipartite multigraph defined directly from the group-algebra boundary map: its vertices index the $R_G$-module bases $(i,j)$ (of $A_1$ and $A_0$), and each group element in the support of a matrix entry $s_e\in\operatorname{supp}\bigl((\partial_A)_{ij}\bigr)$ defines a labeled edge $e_s:j\rightarrow i$. Thus, $\mathcal{P}_A$ is generally a multigraph, since a single matrix entry may contain several supported group elements. The protograph of the BB code in Eq.~\eqref{Eq:sample_BB} is an example of this.

On the other hand, the Tanner graph ${\cal T}(A_\bullet)$ is obtained after binary expansion by replacing each protograph vertex with a fiber indexed by $G$. 
Since $A_1$ and $A_0$ are right $R_G$-modules, an edge $e_s:j\rightarrow i$ labeled by $s_e$ connects $(j,g)\longrightarrow(i,s_e g)$ for every $g\in G$. The Tanner graph obtained in this way is called the \emph{right $G$-derived graph}. 
The same idea applies to $B$ as well.

We now consider the group extension $\widetilde{A}_\bullet$ . The protograph $\mathcal{P}_A$ is kept fixed, while every edge label $s_e\in G$ is replaced by the lifted label
\begin{align}\label{Eq:LiftedProtographVoltage}
    t_e
    =
    \chi(\gamma_e)\sigma(s_e)
    \in H,
    \qquad
    \gamma_e\in K.
\end{align}
Thus, the Tanner graph of $\widetilde{A}_\bullet$ is the right $H$-derived graph of $\mathcal{P}_A$, and index $|K|$-lift of ${\cal T}(A)$.  Similarly, the Tanner graph of $\tilde{B}_\bullet$ can be understood as the \emph{left} $H$-derived graph of the corresponding protograph, as $B_1$ and $B_0$ are left $R_G$-modules

% The construction can therefore be summarized as
% \begin{align}
%     \mathcal{P}_A
%     &\xrightarrow{\;G\text{-derived}\;}
%     \mathcal{T}(A),
%     &
%     \mathcal{P}_A
%     &\xrightarrow{\;H\text{-derived}\;}
%     \mathcal{T}(\widetilde{A}),
% \end{align}
% while leaving $\mathcal{T}(\widetilde{A})$ an index $|K|$-lift of $\mathcal{T}(A)$. 

Different choices of the kernel labels $\{\gamma_e\}$ in Eq.~\eqref{Eq:LiftedProtographVoltage} can produce isomorphic lifted classical codes. This gauge redundancy is described by \emph{vertex switching} in voltage-graph theory~\cite{GROSS1977273,Antoncic2014}. In the present setting, some of the $K$-valued lift component $\gamma_e$ are purely a gauge redundancy, and we can fix them to be the identity element $e_K \in K$. To reduce the search space, we choose spanning forests $F_A \subseteq \mathcal{P}_A$ and $F_B \subseteq \mathcal{P}_B$. A spanning forest contains every protograph vertex and restricts to a spanning tree on each connected component. In particular, it contains no cycles. As shown below, fiber relabeling can set the kernel labels on all forest edges to the identity.

\begin{proposition}[Gauge fixing]\label{Prop:GaugePermutationVoltage}
    Let $F_A$ and $F_B$ be spanning forests of the protographs $\mathcal{P}_A$ and $\mathcal{P}_B$, respectively. Every assignment of kernel labels on these protographs is equivalent, up to fiber relabeling, to an assignment satisfying $\gamma_e=e_K$ for every $e\in F_A\cup F_B$. The original and gauge-fixed
    constituent complexes are isomorphic as $R_H$-module chain complexes,
    and their balanced products define isomorphic LP codes.
\end{proposition}
\begin{proof}
    \textbf{(1).}
    We first discuss gauge fixing. Consider the right-module complex $\widetilde{A}_\bullet$. Assign an element $q_v^A \in \ker\pi$ to each vertex $v$ of $\mathcal{P}_A$, and relabel the $H$-fiber over $v$ by  $(v,h)\longmapsto (v,q_v^Ah)$.
    For an edge $e:j\rightarrow i$ with lifted label $t_e^A$, this relabeling changes the edge label to
    \begin{align}\label{Eq:GaugeTransformA}
        {t_e^A}' = q_i^A t_e^A(q_j^A)^{-1}.
    \end{align}
    Choose a root in each connected component of $F_A$ and fix its fiber labeling. Proceeding outward along $F_A$, since $q_i^A,q_j^A\in\ker\pi$ and $\Im \chi\,{=}\,\ker \pi$, one can choose the labeling at each new vertex to remove $\chi(\gamma_e)$ factor in Eq.~\eqref{Eq:LiftedProtographVoltage}:
    \begin{align}
        {t_e^A}' = \sigma(s_e),
    \end{align}
    or equivalently $\gamma'_e=e_K$, on every tree edge. Since a spanning forest contains no cycles, these choices can be made recursively without consistency constraints.
    
    For the left-module complex $\widetilde{B}_\bullet$, independently assign an element $q_u^B \in \ker\pi$ to each vertex $u$ of $\mathcal{P}_B$, and relabel the corresponding fiber by $(u,h) \longmapsto (u,hq_u^B)$.
    For an edge $e:j\rightarrow i$ with lifted label $t_e^B$, this changes the edge label to
    \begin{align}\label{Eq:GaugeTransformB}
        {t_e^B}' = (q_j^B)^{-1}t_e^Bq_i^B.
    \end{align}
    The same recursive argument sets ${t_e^B}'=\sigma(s_e)$, or equivalently $\gamma'_e=e_K$, on every edge of $F_B$.

    \begin{algorithm*}[t]
\caption{Search for lifted codes from group extensions}
\label{alg:extension-lift-search}
\SetKw{Continue}{continue}
\SetKwFunction{Extensions}{GroupExtensions}
\SetKwFunction{EvaluateLift}{EvaluateLift}

\KwIn{A base LP code
$Q=A\otimes_{R_G} B$, with boundary maps
$\partial_A$ and $\partial_B$ over $R_G$;
selected kernels $\mathcal K$;
a search criterion $\mathsf{crit}$;
a distance-sampling budget $N_{\mathrm{dist}}$.}

\KwOut{For each selected group extension, the lifted code with the highest score.}

Write the entries of the two boundary maps as
\[
(\partial_A)_{ij}
=
\sum_{g\in G}a_{ij,g}g,
\qquad
(\partial_B)_{ij}
=
\sum_{g\in G}b_{ij,g}g.
\]

Define the sets of supported terms
\[
\mathcal S_A
=
\{(A,i,j,g):a_{ij,g}=1\},
\qquad
\mathcal S_B
=
\{(B,i,j,g):b_{ij,g}=1\},
\]
and set $\mathcal S=\mathcal S_A\sqcup\mathcal S_B$\;

Set $\mathcal R\gets[\,]$\;

\ForEach{$K\in\mathcal K$}{
    \ForEach{$E=(1\to K\xrightarrow{\chi}H\xrightarrow{\pi}G\to1)
    \in\Extensions(G,K)$}{

        Choose a section $\sigma:G\to H$ satisfying
        $\pi\circ\sigma=\operatorname{id}_G$\;

        Choose a spanning forest in the protograph Tanner graph of each
        constituent classical code\;

        Let $\Gamma_{\mathrm{gf}}(Q,E,\sigma)\subseteq K^{\mathcal S}$
        be the set of normalized kernel-label assignments satisfying
        $\gamma_e=e_K$ on every spanning-forest edge\;

        Set $\mathsf{best}\gets\varnothing$\;

        \ForEach{$\gamma=(\gamma_A,\gamma_B)
        \in\Gamma_{\mathrm{gf}}(Q,E,\sigma)$}{

            Lift every supported group element in $\partial_A$ and
            $\partial_B$ according to
            \[
            g\longmapsto
            \chi\!\left(\gamma_A(i,j,g)\right)\sigma(g)
            \quad\text{or}\quad
            g\longmapsto
            \chi\!\left(\gamma_B(i,j,g)\right)\sigma(g),
            \]
            respectively\;

            Construct the lifted boundary maps
            \[
            (\widetilde{\partial}_A)_{ij}
            =
            \sum_{g\in G}
            a_{ij,g}\,
            \chi\!\left(\gamma_A(i,j,g)\right)\sigma(g),
            \]
            \[
            (\widetilde{\partial}_B)_{ij}
            =
            \sum_{g\in G}
            b_{ij,g}\,
            \chi\!\left(\gamma_B(i,j,g)\right)\sigma(g),
            \]
            using the same left- and right-module conventions as in the
            base code\;

            Let $\widetilde A$ and $\widetilde B$ be the classical codes
            defined by $\widetilde{\partial}_A$ and
            $\widetilde{\partial}_B$, and construct
            \[
                \widetilde Q
                \gets
                \widetilde A\otimes_{R_H}\widetilde B.
            \]

            $\mathsf{cand}\gets
            \EvaluateLift(\widetilde Q,N_{\mathrm{dist}})$\;
            \tcp{\EvaluateLift constructs the binary CSS matrices, computes $n$ and $k$, 
            and estimates $\widetilde d_X$ and $\widetilde d_Z$.}

            \If{$\mathsf{cand}=\bot$}{
                \Continue\;
            }

            \If{$\mathsf{best}=\varnothing$ or
            $\mathsf{cand}\succ_{\mathsf{crit}}\mathsf{best}$}{
                $\mathsf{best}\gets\mathsf{cand}$\;
            }
        }

        \If{$\mathsf{best}\neq\varnothing$}{
            Append $(E,\mathsf{best})$ to $\mathcal R$\;
        }
    }
}

\Return{$\mathcal R$}\;
\end{algorithm*}

\begin{table*}[!t]
        \centering
    \begin{tabular}{c|c|c|c}
    Group & $\partial_A$ & $\partial_B$ & $\llbracket n,k,d\rrbracket$ \\
    \midrule
    
    $\Z_3\times \Z_3 = \langle x,y \mid x^3,y^3,[x,y]\rangle$
    &
    $\begin{pmatrix}
    1+x+x^2
    \end{pmatrix}$
    &
    $\begin{pmatrix}
    1+y+y^2
    \end{pmatrix}$
    &
    $\llbracket 18,8,2\rrbracket$
    \\
    
    $\Z_6\times \Z_3 = \langle x,y \mid x^6,y^3,[x,y]\rangle$
    &
    $\begin{pmatrix}
    1+x^5+x^4
    \end{pmatrix}$
    &
    $\begin{pmatrix}
    1+x^3y+y^2
    \end{pmatrix}$
    &
    $\llbracket 36,8,4\rrbracket$
    \\
    
    $\Z_9\times \Z_3 = \langle x,y \mid x^9,y^3,[x,y]\rangle$
    &
    $\begin{pmatrix}
    1+x+x^2
    \end{pmatrix}$
    &
    $\begin{pmatrix}
    1+y^2+yx^3
    \end{pmatrix}$
    &
    $\llbracket 54,8,6\rrbracket$
    \\
    $\Z_3\times \Z_3\times \Z_3 = \langle x,y,z \mid x^3,y^3,z^3,[x,y],[x,z],[y,z]\rangle$
    &
    $\begin{pmatrix}
    1+x^2+xz
    \end{pmatrix}$
    &
    $\begin{pmatrix}
    1+y^2+yz
    \end{pmatrix}$
    &
    $\llbracket 54,8,6\rrbracket$
    \\
    $\Z_6\times \Z_6 = \langle x,y \mid x^6,y^6,[x,y]\rangle$
    &
    $\begin{pmatrix}
    1+x^2y^3+x
    \end{pmatrix}$
    &
    $\begin{pmatrix}
    1+y^5+x^3y^4
    \end{pmatrix}$
    &
    $\llbracket 72,12,6\rrbracket$
    \\
    
    $\Z_{12}\times \Z_3 = \langle x,y \mid x^{12},y^3,[x,y]\rangle$
    &
    $\begin{pmatrix}
    1+x^8+x
    \end{pmatrix}$
    &
    $\begin{pmatrix}
    1+x^3y+y^2
    \end{pmatrix}$
    &
    $\llbracket 72,8,8\rrbracket$
    \\
    
    $\Z_{15}\times \Z_3 = \langle x,y \mid x^{15},y^3,[x,y]\rangle$
    &
    $\begin{pmatrix}
    1+x^{10}+x^{11}
    \end{pmatrix}$
    &
    $\begin{pmatrix}
    1+y^2+yx^3
    \end{pmatrix}$
    &
    $\llbracket 90,8,10\rrbracket$\\
    $\Z_{9}=\langle x \mid x^{9}\rangle$
        &
        $\begin{pmatrix}
        1+x^2 & x+x^{5} \\
        1+x^{4} & x^2+x^4
        \end{pmatrix}$
        &
        $\begin{pmatrix}
        1+x^{2} & x+x^{2} \\
        1+x & 1+x
        \end{pmatrix}$
        &
        $\llbracket 72,8,8\rrbracket$
    \\
    $\Z_{15}=\langle x \mid x^{15}\rangle$
    &
    $\begin{pmatrix}
    1+x & x^5+x^{13} \\
    1+x^{8} & x^7+x^8
    \end{pmatrix}$
    &
    $\begin{pmatrix}
    1+x^{13} & x+x^{5} \\
    1+x^{11} & x^3+x^5
    \end{pmatrix}$
    &
    $\llbracket 120,8,13\rrbracket$
    \\
    $\Z_{15}=\langle x \mid x^{15}\rangle$
    &
    $\begin{pmatrix}
    x^3+x^{14} & x^3+x^2 \\
    x^3+x^2 & x^{13}+x^2
    \end{pmatrix}$
    &
    $\begin{pmatrix}
    1+x^8 & x^9+x^7 \\
    1+x^2 & x^5+x^{12}
    \end{pmatrix}$
    &
    $\llbracket 120,8,14\rrbracket$
    \\
    $\Z_{21}=\langle x \mid x^{21}\rangle$
    &
    $\begin{pmatrix}
    x^{15} + x^{19} & x^{15} + x^{17}\\
    1 + x^{2} & x^{3} + x^{7}
    \end{pmatrix}$
    &
    $\begin{pmatrix}
    x^9+x^{19} & 1+x \\
    x^{12} + x^{13} & x^{6}+x^{18} 
    \end{pmatrix}$
    &
    $\llbracket 168,8,15\rrbracket$
    \\
    $\Z_{9}=\langle x \mid x^{9}\rangle$
    &
    $\begin{pmatrix}
    x+x^{2} & 0 & 1 + x^4 \\
    x+x^{5} & 1 + x & 0\\
    0 & 1 + x^5 & 1 + x^2
    \end{pmatrix}$
    &
    $\begin{pmatrix}
    1+x & 1+x^{2} & 0 \\
    0 & x+x^{2} & 1+x^2\\
    1 + x^2 & 0 & x + x^2
    \end{pmatrix}$
    &
    $\llbracket 162,18,8\rrbracket$
    \\
    \midrule
    \end{tabular}
\caption{\textbf{The boundary maps of some of the LP codes presented in the main text.}}
\label{table:boundary_maps}
\end{table*}

    \vspace{5pt} \noindent \textbf{(2).} It remains to show that the gauge fixing also induces an isomorphism of the corresponding balanced-product LP codes.
    
    Let $\widetilde{Q}$ and $\widetilde{Q}'$ denote the LP codes obtained from the classical complexes before and after gauge fixing, respectively. A vertex in any chain degree is labeled by a triple $(v,u,h)$
    where $v$ and $u$ index basis elements of the relevant $\widetilde{A}$- and $\widetilde{B}$-modules, respectively, and $h\in H$. The Tanner graph of $\widetilde{Q}$ contains two types of edges. An $A$-type edge has the form
    \begin{align}\label{Eq:AtypeEdgeinLP}
        (j,u,h)
        \longrightarrow
        (i,u,t_e^Ah),
    \end{align}
    whereas a $B$-type edge has the form
    \begin{align}\label{Eq:BtypeEdgeinLP}
        (v,j,h)
        \longrightarrow
        (v,i,ht_e^B).
    \end{align}
    The corresponding edges in $\widetilde{Q}'$ are defined analogously using ${t_e^A}'$ and ${t_e^B}'$.
    
    Define
    \begin{align}
        \Phi:
        V\bigl(\mathcal{T}(\widetilde{Q})\bigr)
        &\longrightarrow
        V\bigl(\mathcal{T}(\widetilde{Q}')\bigr), \nonumber \\
        (v,u,h)
        &\longmapsto
        \bigl(v,u,q_v^Ahq_u^B\bigr).
    \end{align}
    The map $\Phi$ is a bijection. We show that $\Phi$ preserves the edge incidence relations and thus it is a graph isomorphism.
    
    For the $A$-type edge in Eq.~\eqref{Eq:AtypeEdgeinLP}, one has
    \begin{align}
        \Phi(i,u,t_e^Ah)
        &=
        \bigl(i,u,q_i^At_e^Ahq_u^B\bigr) \nonumber \\
        &=
        \bigl(i,u,{t_e^A}'q_j^Ahq_u^B\bigr),
    \end{align}
    where the second equality follows from Eq.~\eqref{Eq:GaugeTransformA}. By definition, this is the endpoint of the $A$-type edge in $\mathcal{T}(\widetilde{Q}')$ whose initial vertex is $\Phi(j,u,h) = \bigl(j,u,q_j^Ahq_u^B\bigr)$.
    Thus, $\Phi$ preserves all $A$-type edges.
    
    Similarly, for the $B$-type edge in Eq.~\eqref{Eq:BtypeEdgeinLP},
    \begin{align}
        \Phi(v,i,ht_e^B)
        &=
        \bigl(v,i,q_v^Aht_e^Bq_i^B\bigr) \nonumber \\
        &=
        \bigl(v,i,q_v^Ahq_j^B{t_e^B}'\bigr),
    \end{align}
    where Eq.~\eqref{Eq:GaugeTransformB} was used in the second equality. This is the endpoint of the $B$-type edge in $\mathcal{T}(\widetilde{Q}')$ beginning at $\Phi(v,j,h) = \bigl(v,j,q_v^Ahq_j^B\bigr)$.
    Hence, $\Phi$ also preserves all $B$-type edges.
    
    Therefore, $\Phi$ is a graph isomorphism $\mathcal{T}(\widetilde{Q}) \simeq \mathcal{T}(\widetilde{Q}')$. The two LP codes are consequently related by permutations of the qubits and of the $X$- and $Z$-check generators and are therefore isomorphic.
\end{proof}

We note that the spanning forests in Proposition~\ref{Prop:GaugePermutationVoltage} belong to the protographs $\mathcal{P}_A$ and $\mathcal{P}_B$, rather than to the derived Tanner graphs $\mathcal{T}(A)$ and $\mathcal{T}(B)$.

Let $V_A,E_A$ and $V_B,E_B$ denote the vertex and edge sets of the two protographs, and let $c_A$ and $c_B$ denote their numbers of connected components. Before gauge fixing, there are $|K|^{|E_A|+|E_B|}$ possible kernel-label assignments. A spanning forest of $\mathcal{P}_A$ contains $|V_A|-c_A$ edges, while one of $\mathcal{P}_B$ contains $|V_B|-c_B$ edges. Gauge fixing these labels therefore reduces the search space to
\begin{align}
    |K|^{
        |E_A|-|V_A|+c_A
        +
        |E_B|-|V_B|+c_B
    }.
\end{align}
For the BB code in Eq.~\eqref{Eq:sample_BB}, each constituent protograph consists of two vertices connected by three parallel edges. One edge can therefore be gauge-fixed in each protograph. The combined search space is reduced from $|K|^6$ to $|K|^4$. More generally, if both constituent protographs are connected, the overall search space is reduced by a factor of $|K|^{n_0^A + n_1^A + n^B_0 + n^B_1-2}$, where $A_i \simeq R_G^{ n_i^A}$, and similarly for $B_i$.

\subsection{Boundary map of LP codes}\label{App:boundarymaps}
Here we provide the boundary maps for selected lifted codes, summarized in Table~\ref{table:boundary_maps}. We provide the boundary maps for simple cyclic groups where the ring admits a compact polynomial representation. For other codes, the boundary map representation using the \texttt{GAP} representation will be disclosed in a GitHub repository.

\subsection{The canonical logical operators in the gross code}\label{App:GrossLogical}
We first give explicit representatives of logical operators in the gross code $\llbracket144, 12, 12\rrbracket$. Let us first define the polynomials
\begin{align}\nonumber
    f &= 1 + x + x^2 + x^3 + x^6 + x^7 \\
    &\quad+ x^8 + x^9 + xy^3 + x^5y^3 +x^7y^3 + x^{11} y^3\\
    g &= x + x^2y + y^2 + xy^2 + x^2y^3 + y^4\\
    h&= 1 + y + xy + y^2 + y^3 + xy^3,
\end{align}
and a set of monomials
\begin{align}\label{Eq:GrossLogMonomialsA}
    A &= \{1, y, x^2y, x^2y^5, x^3y^2, x^4\}\\\label{Eq:GrossLogMonomialsB}
    B &= \{y, y^5, xy, 1, x^4, x^5y^2\}.
\end{align}
We can take the logical operators as
\begin{align}
    \overline{X}_i = \begin{pmatrix}
        a_i f\\
        0
    \end{pmatrix}, 
    \quad \overline{Z}_j = \begin{pmatrix}
        b_jh^\ast\\
        b_jg^\ast
    \end{pmatrix},
\end{align}
and
\begin{align}
    \overline{X}_{i + 6} = \begin{pmatrix}
        a_i g\\
        a_i h
    \end{pmatrix}, 
    \quad \overline{Z}_{j + 6} = \begin{pmatrix}
        0\\
        b_jf^\ast
    \end{pmatrix},
\end{align}
where $a_i \;[b_j]$ is the $i$-th [$j$-th] element of $A\;[B]$ in Eq.~\eqref{Eq:GrossLogMonomialsA} [\eqref{Eq:GrossLogMonomialsB}] for $i = 1, \dots, 6$ and $j = 1, \dots, 6$. These operators span a canonical basis of the logical operators of the gross code; i.e. $\overline{X}_k$ and $\overline{Z}_{k'}$ anticommute iff $k = k'$ for $k = 1, \dots, 12$ and $k' = 1,\dots, 12$.

\section{Calculation of coherent information}\label{App:SecCI}

We calculate the coherent information of a CSS code under independent $X$ errors occurring with probability $p$ on each qubit. An error configuration is represented by $\hat{\eta}\in\F_2^n$, where $\hat{\eta}_i=1$ indicates an error on qubit $i$. Its probability is
\begin{align} 
P(\hat{\eta})
&=
p^{\operatorname{wgt}(\hat{\eta})}
(1-p)^{n-\operatorname{wgt}(\hat{\eta})},
\end{align}
where $\operatorname{wgt}(\hat{\eta})$ denotes the Hamming weight of $\hat{\eta}$.

For a CSS code encoding $k$ logical qubits, the coherent information under $X$ errors can be written as~\cite{hlfh-86yz, PhysRevA.111.032402}
\begin{align}\label{Eq:CIunderDecoherence}
I_C
&=
k+
\sum_{\hat{\eta}\in\F_2^n}
P(\hat{\eta})
\log_2
\frac{
\sum_{\hat{v}\in\operatorname{Im}(H_X^T)}
P(\hat{\eta}+\hat{v})
}{
\sum_{\hat{u}\in\ker(H_Z)}
P(\hat{\eta}+\hat{u})
},
\end{align}
where $H_X$ and $H_Z$ are the $X$- and $Z$-check matrices, respectively, and all vector additions are taken over $\F_2$.

Let $L_X$ be a matrix whose rows represent a basis of logical $X$ operators, or equivalently a basis of the quotient $\ker(H_Z)/\operatorname{Im}(H_X^T)$. Choosing representatives of these logical classes gives
\begin{align}
\ker(H_Z)
&=
\operatorname{Im}(H_X^T)
\oplus
\operatorname{Im}(L_X^T).
\end{align}
If $r_X$ denotes the number of rows of $H_X$, every element of $\ker(H_Z)$ can therefore be expressed as
\begin{align}
\hat{u}
&=
H_X^T\hat{\tau}_C
+
L_X^T\hat{\tau}_L,
&
\hat{\tau}_C
&\in\F_2^{r_X},
&
\hat{\tau}_L
&\in\F_2^k.
\end{align}
When the rows of $H_X$ are linearly dependent, this parametrization is not unique. However, every element of $\operatorname{Im}(H_X^T)$ occurs with the same multiplicity, which cancels from the ratios below.

For a fixed error configuration $\hat{\eta}$, define the partition function of the logical sector $\hat{\ell}\in\F_2^k$ by
\begin{align} 
Z_{\hat{\eta}}(\hat{\ell})
&=
\sum_{\hat{\tau}_C\in\F_2^{r_X}}
P\left(
\hat{\eta}
+
H_X^T\hat{\tau}_C
+
L_X^T\hat{\ell}
\right).
\end{align}
The associated normalized logical-sector probability is
\begin{align}  \label{eq:logical_dist}
q_{\hat{\eta}}(\hat{\ell})
&=
\frac{
Z_{\hat{\eta}}(\hat{\ell})
}{
\sum_{\hat{\ell}'\in\F_2^k}
Z_{\hat{\eta}}(\hat{\ell}')
}.
\end{align}
Equation~\eqref{Eq:CIunderDecoherence} then becomes
\begin{align} 
I_C
&=
k+
\sum_{\hat{\eta}\in\F_2^n}
P(\hat{\eta})
\log_2 q_{\hat{\eta}}(\bm{0}).
\end{align}

We evaluate $q_{\hat{\eta}}(\bm{0})$ by mapping the sector partition functions to a classical spin model. Define
\begin{align}
M
&=
\left[
H_X^T
\middle|
L_X^T
\right], \quad \hat{\tau} =
\begin{pmatrix}
\hat{\tau}_C\\
\hat{\tau}_L
\end{pmatrix}
\in\F_2^{r_X+k}.
\end{align}
Then $M\hat{\tau} = H_X^T\hat{\tau}_C + L_X^T\hat{\tau}_L$. 
We associate binary variables $\hat{\cdot}$ with Ising variables ${\cdot}$ according to
\begin{align}
\eta_i
&=
1-2\hat{\eta}_i, \quad \tau_j = 1-2\hat{\tau}_j,
\end{align}
so that $\eta\in\{\pm1\}^n$ and $\tau\in\{\pm1\}^{r_X+k}$. For a fixed disorder realization $\eta$, consider the classical Hamiltonian
\begin{align}\label{Eq:MCHam}
\beta\mathcal H_{\eta}(\tau)
&=
-\beta J
\sum_{i=1}^{n}
\eta_i
\prod_{j:M_{ij}=1}
\tau_j,
\end{align}
with the Nishimori coupling
\begin{align}
\beta J
&=
\frac{1}{2}
\log\frac{1-p}{p}.
\end{align}
Using 
\begin{align}
    (-1)^{[\hat{\eta} + M \hat{\tau}]_i} = \eta_i \prod_{j: M_{ij} = 1} \tau_j,
\end{align}
we obtain
\begin{align}\label{P_error_given_syndrome}
P(\hat{\eta}+M\hat{\tau})
&= \bigl[p(1-p)\bigr]^{n/2}
e^{-\beta\mathcal H_{\eta}(\tau)}.
\end{align}
Thus, the error probability is proportional to the Boltzmann weight of the classical model, with a prefactor independent of $\tau$.

For each sampled error configuration $\hat{\eta}$, we perform an MCMC simulation of Eq.~\eqref{Eq:MCHam}. After thermalization, each sampled Ising configuration is converted to its binary representation $\hat{\tau}=\hat{\tau}_C\oplus\hat{\tau}_L$. The marginal distribution of $\hat{\tau}_L$ is precisely the logical-sector distribution in Eq.~\eqref{eq:logical_dist}. In particular,
\begin{align}
q_{\hat{\eta}}(\bm{0})
&=
\Pr_{\eta}\left(\hat{\tau}_L=\bm{0}\right).
\end{align}
If $N_{\mathrm{MC}}$ configurations are retained after thermalization and $N_0(\hat{\eta})$ of them satisfy $\hat{\tau}_L=\bm{0}$, then
\begin{align} 
\widehat q_{\hat{\eta}}(\bm{0})
&=
\frac{N_0(\hat{\eta})}{N_{\mathrm{MC}}}
\end{align}
estimates the ratio in Eq.~\eqref{eq:logical_dist}. Finally, drawing $N_{\eta}$ independent error configurations gives
\begin{align}\label{Eq:CIEstimator}
\widehat I_C
&=
k+
\frac{1}{N_{\eta}}
\sum_{a=1}^{N_{\eta}}
\log_2
\widehat q_{\hat{\eta}^{(a)}}(\bm{0}).
\end{align}
Convergence must be checked with respect to the thermalization time, the number of MCMC samples, and the number of disorder realizations. Because Eq.~\eqref{Eq:CIEstimator} contains the logarithm of an empirical frequency, the zero logical sector must be sampled sufficiently often for finite-sample bias to be controlled.

\bibliography{manuscript}

\end{document}